\documentclass[sigconf]{acmart}
\AtBeginDocument{%
  }

\copyrightyear{2026}
\acmYear{2026}
\setcopyright{cc}
\setcctype{by}
\acmConference[CCS '26] {Proceedings of the 2026 ACM SIGSAC Conference on Computer and Communications Security}{November 15--19, 2026}{The Hague, Netherlands.}
\acmBooktitle{Proceedings of the 2026 ACM SIGSAC Conference on Computer and Communications Security (CCS '26), November 15--19, 2026, The Hague, Netherlands}
\acmISBN{979-8-4007-2871-6/2026/11}
\acmDOI{10.1145/XXXXXX.XXXXXX}

\usepackage{amsmath}
\usepackage{kotex}
\usepackage{url}
\usepackage{xspace}

\usepackage{amsthm}
\usepackage{xcolor}
\usepackage{dsfont}
\usepackage{enumitem}
\usepackage{multirow}

\newtheorem{definition}{Definition}

\usepackage{balance}

\usepackage{threeparttable}
\usepackage[table]{xcolor}

\usepackage{algorithm, algorithmic}

\newcommand{\net}{\textbf{\textsc{net}}\xspace}
\begin{document}

% \title{
% Casting the Net! Impersonation Attack Against Biometric Authentication Systems Under Realistic Threat Model}
\title{
Casting the Net! Revisiting MasterFace Impersonation Attacks
}

% \author{Anonymous}

% \affiliation{%

\author{Seunghun Paik}
\authornote{Department of Mathematics \& Research Institute for Natural Sciences}
\authornote{Equally Contributed}
\affiliation{%
  % \institution{Research Institute for Natural Sciences}
  \institution{Hanyang University}
  \city{Seoul}
  \country{Republic of Korea}  
}
\email{whitesoonguh@hanyang.ac.kr}
\author{Sunpill Kim}
\authornotemark[1]
\authornotemark[2]
\affiliation{%
  % \institution{Research Institute for Natural Sciences}
  \institution{Hanyang University}
  \city{Seoul}
  \country{Republic of Korea}    
}
\email{ksp0352@hanyang.ac.kr}
\author{Chanwoo Hwang}
\authornotemark[1]
\affiliation{%
  % \institution{Research Institute for Natural Sciences}
  \institution{Hanyang University}
  \city{Seoul}
  \country{Republic of Korea}    
}
\email{aa5568@hanyang.ac.kr}
\author{Jae Hong Seo}
\authornotemark[1]
\authornote{Corresponding Author}
\affiliation{%
  % \institution{Research Institute for Natural Sciences}
  \institution{Hanyang University}
  \city{Seoul}
  \country{Republic of Korea}    
}
\email{jaehongseo@hanyang.ac.kr}

% \affiliation{%

% \affiliation{%

% \affiliation{%

% \affiliation{%

% \affiliation{%

% \affiliation{%

% \renewcommand{\shortauthors}{S. Paik et al.}

\begin{abstract}

Impersonation is a fundamental security threat in face recognition systems (FRSs).
While the security of FRSs has been challenged by various attack vectors, under realistic adversarial capabilities, e.g., a limited number of decision-only authentication trials and no internal system knowledge, most attack techniques become infeasible.
As a result, impersonation by \textit{zero-effort impostors}, characterized by false match rate (FMR), is commonly regarded as a standalone baseline. 
A few years ago, impersonation attacks based on MasterFaces---which exploit non-uniformity in biometric distribution---emerged as a notable security threat that could break the barrier of the FMR-based baseline under such realistic constraints.
However, they were believed not to yield impersonation above the standard FMR in modern FRSs, as discussed by multiple follow-up studies.

In this paper, we demonstrate that even legitimate access to public commercial APIs allows an adversary to amplify impersonation rates through MasterFaces, resulting in a non-trivial impersonation attack beyond FMR on downstream applications built on top of these APIs.
We observe that several real-world FRS deployments are implemented using commercial APIs, and that the backend service provider is publicly disclosed or trivially inferable.
As a result, the adversary can purchase these \textit{pay-as-you-go} API services without requiring any additional privilege over the target FRS.
Motivated by this observation, we formalize the MasterFaces attack as a maximum coverage problem over the biometric representation space, which we call a \net, and show that the adversary can construct an API-tailored \net by leveraging the geometric structure of the representation space.
Through experiments, we demonstrate that our attack amplifies the impersonation rates of several open-source and commercial API-based FRSs by up to 9.5$\times$ within at most 30 authentication trials, compared to those expected from the standard FMR.
Overall, we revive the MasterFaces attack as posing a potential vulnerability against real-world FRSs.
\end{abstract}


% \ccsdesc[300]{Security and privacy~Human and societal aspects of security and privacy}



\maketitle

\section{Introduction}

Face recognition systems (FRSs)~\cite{zhao2003face, schroff2015facenet, deng2019arcface} have been widely adopted in real-world applications for seamless and user-friendly authentication, such as access control in buildings~\cite{DHSBACS} or airports~\cite{starAlliance}, and identity verification in financial services~\cite{FIDO2}.
In many such deployments, FRSs are implemented using commercial API services, e.g., AWS Rekognition~\cite{rekognition} and Tencent Cloud~\cite{tencent}, which serve as a backbone of large-scale real-world systems~\cite{AWSUseCases, TencentUseCases}.
Standards for FRSs have also been well established, including ISO/IEC documents for how to evaluate and analyze these systems in various scenarios~\cite{iso2021iso19795, iso2024iso24741, iso2023iso30107}, as well as benchmarks~\cite{NISTFRVT}.

The security of FRSs has been challenged by a wide range of attacks under various threat models~\cite{sharif2016accessorize, dong2019efficient, kim2024scores, kim2025non, li2023sibling, shahreza2024vulnerability, galbally2010vulnerability, jeong2022analysis, damer2018morgan, colbois2023approximating, an2023imu}.
Among these threats, the fundamental adversarial goal as an authentication system is \textit{impersonation}, where an adversary aims to be accepted as a victim identity enrolled in the system.
The risk against impersonation is typically captured by false match rate (FMR)~\cite{grother2013biometric}, a standard evaluation metric that quantifies the probability that two facial images from different users are incorrectly accepted as a match.
This reflects the \textit{zero-effort imposter}, which assumes no deliberate attempt at impersonation and thus represents the weakest adversary, serving as the foundational security baseline for FRSs.
Accordingly, several standards and benchmarks place a heavy emphasis on measuring and reporting FMR, and often discuss typical operating points for different application contexts~\cite{grother2013biometric}.

In realistic FRS deployments, the adversarial capability is highly constrained due to operational and security considerations.
Adversaries typically have no access to internal system information, e.g., model parameters or prior knowledge of the enrolled identity.
In addition, only a highly restricted form of interaction with the system is permitted, such as obtaining decisions from a small, fixed number of authentication trials.
These constraints substantially limit many attack strategies, which rely on high-fidelity information from the system through rich interaction~\cite{mai2018reconstruction, kim2024scores, shahreza2024vulnerability, galbally2010vulnerability, jeong2022analysis, sharif2016accessorize, an2023imu} or privileged access to the target identity or enrollment~\cite{dong2019efficient, kim2025non, li2023sibling, damer2018morgan, colbois2023approximating}.
As a result, under such a restrictive setting, the aforementioned zero-effort impostors have served as a baseline, and they are widely regarded as capturing the practical security level of modern FRS deployments, as well as biometric authentication systems and derived security protocols~\cite{katsumata2021revisiting, boldyreva2025may, bauspiess2024brake}.

A few years ago, MasterFace-based impersonation attacks~\cite{nguyen2020generating, shmelkin2021generating} challenged this conventional understanding by exploiting non-uniformity in biometric distributions.
They demonstrated that carefully synthesized faces---called MasterFaces---can be matched with multiple identities and exploited for impersonation by facilitating false matches.
Although these attacks received considerable attention at that time, multiple follow-up studies showed that they were effective for outdated FRSs only; for recent state-of-the-art FRSs under a realistic threat model, they could not achieve impersonation success rates beyond the FMR~\cite{terhorst2022limited, nguyen2022master, friedlander2022generating}.
Thus, zero-effort impostors have remained the de facto baseline adversary, and to the best of our knowledge, it remains open whether it is possible to design a non-trivial impersonation attack under this setting.

\begin{figure*}[t]
    \centering
    \includegraphics[width=\linewidth]{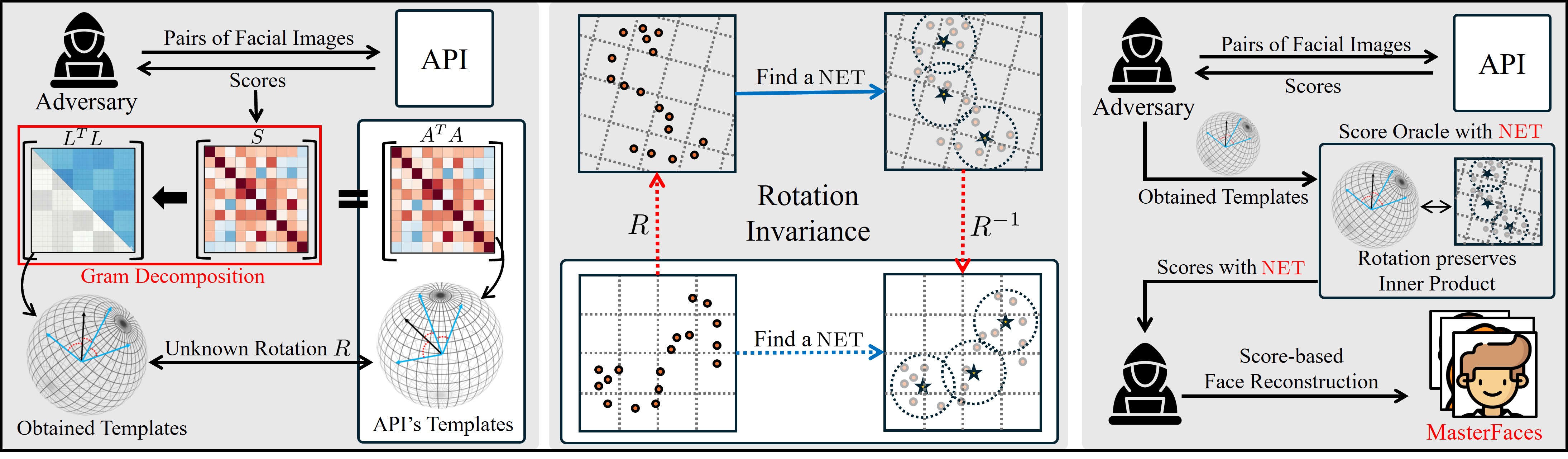}
    \vspace{-22pt}
    \caption{
    Overview of our attack pipeline. 
    Left: The adversary obtains rotated templates via the Gram Decomposition. 
    Middle: The adversary can obtain a rotated \net due to its rotation invariance. 
    Right: The adversary finally recovers MasterFaces from the found \net via a score-based reconstruction method. 
    White regions are not accessible to the adversary. Best viewed in color.
    }
    \label{fig:overview}
    \vspace{-7pt}
\end{figure*}

\subsection{Our Contribution}
In this paper, we show that the adversary can \textit{legitimately} leverage public commercial APIs to amplify the impersonation capability of MasterFaces, thereby enabling a non-trivial impersonation attack beyond the zero-effort impostor baseline under the same authentication trial budget.
Here, ``legitimately'' means that the adversary purchases the public API and issues queries to obtain similarity scores between each pair of chosen facial images, without requiring any additional privilege on the target FRS or auxiliary information on the enrolled identity.
We observe that several FRS services utilize commercial APIs as a backend, such as Amazon Rekognition~\cite{rekognition} or Tencent Cloud~\cite{tencent}, and it is typically publicly disclosed which backend API was used~\cite{AWSUseCases, TencentUseCases}.
As a result, the adversary can purchase the same \textit{pay-as-you-go} API\footnote{\url{https://aws.amazon.com/rekognition/pricing/}} with the target FRS.

Motivated by this observation, we design a novel algorithm to find MasterFaces tailored to a target FRS.
We formalize MasterFaces as the samples whose corresponding templates---i.e., feature embeddings extracted by an FRS---collectively maximize coverage over the face template distribution.
We call this set of templates a \net\footnote{This is an analogy to casting the net at the best plausible spot for fishing.}.
This formalization allows us to exploit geometric properties of the template space, including isometries, and reduces \net construction to a maximum coverage problem (MCP), a classical combinatorial optimization problem whose polynomial-time approximation is well-known~\cite{nemhauser1978analysis}. 
Furthermore, we show that the adversary can obtain templates of chosen facial images from a public API up to an unknown isometry.
We find that these transformed features suffice to construct an API-tailored \net, and more importantly, faces corresponding to \net can be reconstructed using existing score-based reconstruction attacks~\cite{kim2024scores} as a subroutine.

We implement our attack and conduct extensive experimental analyses against various open-source and Amazon Rekognition API-based FRSs.
When the authentication trial budgets range from 5 to 30, reflecting typical rate limits in real-world deployments, our attack shows consistently higher impersonation rates compared to the baseline zero-effort impostor by 2.03$\times$--8.17$\times$ and 1.75$\times$--9.50$\times$ for the FMR of $10^{-4}$ and $10^{-5}$, respectively.
Note that the adversary can launch our attack within a realistic API usage budget, spending at most \$100 according to their pricing policy.
Our results highlight that MasterFaces-based impersonation attacks can be revived through the legitimate exploitation of publicly available information about backend API usage, underscoring the need for greater caution in system-level information disclosure.

We summarize our contribution as follows:
\begin{itemize}
    \item We present a non-trivial impersonation attack beyond the standard FMR baseline, showing that MasterFaces become effective when the adversary can legitimately purchase the same public commercial APIs as the target FRS.
    
    \item We formalize MasterFaces by constructing a \net, which allows for reducing the attack to the maximum coverage problem, leveraging the geometry of the template space.

    \item We show how to construct an API-tailored \net from legitimate API score queries, whose corresponding faces can be recovered from existing score-based reconstruction attacks.
    
    \item We implement our attack against multiple open-source and commercial API-based FRSs, demonstrating consistently higher impersonation rates compared to the zero-effort impostor baseline under authentication trials ranging from 5 to 30.

    \item We provide our source code on the Zenodo archive \url{https://zenodo.org/records/20765343} for reproducibility.
    To avoid potential misuse, we release only a minimal part of the code.
\end{itemize}

\subsection{Technical Overview}

MasterFace attacks aim to find a face that can be identified as many different identities as possible, expecting that such a face would be likely to succeed in impersonation as an unknown target identity.
However, previous methods search MasterFaces over the latent space of the generative models~\cite{karras2019style}, which often results in sub-optimal solutions.
Subsequent studies further showed that such MasterFaces heavily rely on the FRS and would not be effective for modern sophisticated FRSs~\cite{nguyen2022master, terhorst2022limited, friedlander2022generating}.
Consequently, against unknown modern FRSs, their MasterFaces were barely effective, i.e., no better than the zero-effort impostor baseline~\cite{terhorst2022limited}.

To address these limitations, our attack consists of two core ingredients: \net-based formulation of MasterFaces and exploiting legitimate API queries for crafting API-tailored MasterFaces.
We visualize an overview of our attack's pipeline in Fig.~\ref{fig:overview}.

\subsubsection*{\bf Revisiting MasterFaces over Template Space}
Instead of directly finding MasterFaces over the face image space, we consider the template space, where we can leverage its geometric structure.
From this perspective, we can observe that the corresponding templates of MasterFaces should maximize the coverage of the face templates from the given dataset.
We focus on finding such a set of vectors over the template space, which we define as a \net.
Finding a \net is equivalent to solving the MCP, and this \net would serve a role as the optimal MasterFaces over the template space.

\subsubsection*{\bf Exploiting Legitimate API Queries}
To proceed with the impersonation attack via \net, the adversary must recover the corresponding facial images.
Notably, \net itself still heavily relies on the template space, so faces corresponding to \net may not be effective in an unknown target FRS.
We observe that, if the adversary can purchase the backend API of the target FRS, it becomes possible to construct an API-tailored \net and recover corresponding MasterFaces, thereby addressing all these issues.

Such APIs typically return a confidence score for a pair of queried facial images, and prior studies~\cite{knoche2023explainable, kim2024scores, kim2025non} have shown that the corresponding cosine similarity, i.e., an inner product value between unit vectors, can be inferred by the adversary.
From this setting, our key observation is that the adversary can decompose the score matrix, which consists of the inner product values of the templates.
More precisely, if we denote $S$ as the score matrix, then $S=A^{T}A$ for the matrix $A$ whose columns consist of templates extracted by the API.
By applying the Gram decomposition, the adversary obtains $L$ such that $A^{T}A = L^{T}L$. 
A classical linear algebra theory ensures that $L=RA$ for some unknown rotation matrix $R$.
This further allows the adversary to obtain a \textit{rotated} template of the chosen facial image, because for the template $z$ extracted by the API, the adversary can obtain $z^{T}A$ through queries and $z^{T}A=(Rz)^{T}(RA)$ holds; we can deduce that $Rz = (z^{T}A \cdot L^{-1})^{T}$.

Note that the adversary does not know the exact template because the rotation is unknown.
Nevertheless, finding MasterFaces is still possible because the coverage is invariant under isometries.
That is, the found \net from rotated templates is equal to rotating the \net from the original templates by the same amount.
Hence, when applying the Gram decomposition technique again, the adversary can obtain the score between the MasterFaces corresponding to \net and any chosen facial image.
This directly enables the adversary to apply score-based attacks on \net, which reconstruct the face image from score queries with it.
We remark that all these processes remain within our threat model because no further privilege or interaction with the target identity is required.

\section{System and Threat Model}

We first clarify the system and the threat model we are considering.
To this end, we briefly introduce the FRS and the impersonation attack, and then formalize our problem setting with justifications in real-world deployment scenarios.

\subsubsection*{\bf Notation} 
For a set $\mathcal{X}$ and the metric function $d_{\mathcal{X}}: \mathcal{X} \times \mathcal{X} \rightarrow \mathbb{R}_{\ge 0}$, we denote $(\mathcal{X}, d_{\mathcal{X}})$ as the metric space. 
For a distribution $\mathcal{D}$, we denote $x \gets \mathcal{D}$ as the random sampling.
In particular, for a finite set $\mathcal{S}$, we denote $x \xleftarrow{\$} \mathcal{S}$ as the random uniform sampling over $\mathcal{S}$.
We denote $\mathbb{R}^{d-1}$ and $\mathbb{S}^{d-1}$ as $d$-dimensional Euclidean space and $(d-1)$-dimensional hypersphere embedded into $\mathbb{R}^{d}$, respectively.
We denote $\langle \cdot, \cdot \rangle$ as a standard inner product over $\mathbb{R}^{d}$.
We denote a vector as a column-wise, i.e., $d \times 1$ matrix.

\subsection{Backgrounds on FRSs}

\subsubsection*{\bf Face Recognition Systems}
Let $\mathcal{B}$ be the set of faces and $(\mathcal{X}, d_{\mathcal{X}})$ be a metric space. 
We denote a function $F: \mathcal{B} \rightarrow \mathcal{X}$ as a \textit{template extractor}.
In this context, we call each element of $\mathcal{X}$ a face template.
We expect that $F$ preserves the implicit similarity between faces, i.e., facial images from the same identity are mapped to close templates, or vice versa.
With a template extractor $F$, we define a face recognition system (FRS) as a pair of algorithms $(\mathsf{Enroll}, \mathsf{Auth})$ corresponding to enrollment and authentication, respectively.
Each algorithm has the following functionality:
\begin{itemize}
    \item $\mathsf{Enroll}$ takes a facial image $\mathfrak{b} \in \mathcal{B}$ and returns a corresponding template $t := F(\mathfrak{b})$. Later, the FRS stores $t$ in its database.

    \item $\mathsf{Auth}$ takes a facial image $\mathfrak{b}' \in \mathcal{B}$ and the stored template $t$, returning a decision $0/1$ corresponding to ``reject'' and ``accept''. For a pre-determined threshold $\tau > 0$, this is typically implemented by $\mathds{1}(d_{\mathcal{X}}(F(\mathfrak{b}'), t) < \tau)$.
\end{itemize}
For the ease of explanation, at this moment, we consider the FRS as a 1:1 verification system.
Our analysis can be naturally extended to a 1:N identification system, which will be discussed in Section~\ref{sec:extensions}.

We note that many state-of-the-art FRSs~\cite{schroff2015facenet, deng2019arcface, kim2022adaface, kim2024keypoint, dan2024topofr} utilize $\mathcal{X} = \mathbb{S}^{d-1}$ and $d_{\mathcal{X}}$ as the cosine distance, which is defined by $d_{\cos}(x,y) := 1 - \langle x, y \rangle$ for $x, y \in \mathbb{S}^{d-1}$.
Since our primary focus is on such FRSs, we consider the template space $(\mathcal{X}, d_{\mathcal{X}})$ as $(\mathbb{S}^{d-1}, d_{\cos})$ in later sections.

\subsubsection*{\bf False Match Rate and Impersonation Attack}
Informally speaking, the FMR of the given FRS refers to the probability that a pair of biometrics from different identities is determined as the same identity.
Since the authentication algorithm is parametrized by the threshold $\tau$, so is FMR; a smaller $\tau$ gives a smaller FMR.
In standard benchmarks like IJB-C testsuite~\cite{maze2018iarpa}, FMR at the threshold $\tau$ is measured by $\mathsf{FMR}(\tau):= \frac{1}{|\mathcal{P}|}\sum_{(\mathfrak{b},\mathfrak{b}')\in \mathcal{P}} \mathds{1}(d_{\mathcal{X}}(F(\mathfrak{b}), F(\mathfrak{b}')) < \tau)$, where $\mathcal{P}$ consists of pairs of biometrics from different identities in the benchmark set.

For the stored template $t$, we consider the impersonation attack as the problem of finding a facial image $\mathfrak{b}'$ such that $\mathsf{Auth}(\mathfrak{b}', t) = 1$, without knowing the stored template $t$.
The capability of the adversary depends on the attack scenarios; in usual real-world applications, the adversary does not know either the template extractor of the target FRS or the enrolled facial image of the target identity.
In this setting, we can consider a simple impersonation attack that solely relies on FMR, so-called the \textit{zero-effort impostor} mentioned in standards~\cite{grother2013biometric, iso2021iso19795}.
To impersonate the enrolled identity in the target FRS, the adversary selects one of the facial images and tries to authenticate with this sample.
Provided that the threshold used for the benchmark dataset well represents the whole population, as commonly assumed in the standards~\cite{maze2018iarpa}, the adversary's attack success rate is the same as FMR.

\subsubsection*{\bf FRS with Commercial API}
Several real-world FRSs leverage 3rd party API services, e.g., Amazon Rekognition~\cite{rekognition} or Tencent Cloud~\cite{tencent}, in the place of the template extractor.
When querying two facial images to these services, they return a confidence score that represents the similarity of the queried pairs in terms of a value between 0 and 1.
Note that the formula for confidence scores is not publicly available, but the API services provide preset thresholds for the confidence score.
In addition to this, these services allow storing the user's faces on their own cloud server, and only an authorized entity (e.g., FRS service provider) can access the server to obtain a confidence score between the queried faces and the enrolled ones.

\subsection{Problem Formulation}

\subsubsection*{\bf System Model}
We consider two entities: a user and the FRS service provider.
We assume that the user has already enrolled in the FRS service, and the service provider stores the template in its storage.
The functionalities of each party are as follows:
\begin{itemize}
    \item The user queries a facial image to the service provider, receiving a decision ``accept''or ``reject''.

    \item The service provider runs $\mathsf{FRS.Auth}$ on the queried facial image and the enrolled template, and return the decision result to the user.
\end{itemize}
As documented in NIST and ISO/IEC standards~\cite{grother2013biometric, iso2021iso19795}, we assume that there is a maximum number of queries made by the user.
In usual real-world deployments, the maximum number of allowed authentication trials is often strictly limited (e.g., to a small number such as tens of attempts), and the query limit can be refreshed periodically or upon successful authentication.

\subsubsection*{\bf Threat Model}
From the above setting, we consider an adversary that serves the role of the user, aiming to retrieve ``accept'' from the service provider within a limited number of queries.
We assume that the adversary has no prior knowledge of both the target system and the enrolled identity.
The adversary cannot gain access to any internal information, such as architecture, training dataset, and the parameters of the template extractor, the enrolled template stored in the database, and the decision threshold.
Nevertheless, we assume that the adversary can exploit some public face datasets, such as usual datasets for training template extractors~\cite{deng2019retinaface, zhu2021webface260m, karkkainen2021fairface}.

In particular, when the target FRS utilizes a commercial third-party API, regulatory requirements often mandate transparency regarding face data handling practices and, in some cases, the involvement of third-party service providers~\cite{voigt2017eu}.
In addition, the adversary could infer which third-party service is used by referring to publicly available use-cases of these services~\cite{AWSUseCases, TencentUseCases}.
Hence, the adversary can purchase the same API used by the target FRS and query pairs of its own facial images without knowing the target identity.
These API services are \textit{pay-as-you-go}, hence the adversary can make many more queries to the API service than direct authentication trials to the FRS.
Importantly, we note that access to such APIs does not grant any additional privileges over the target FRS.

To reflect the real-world application scenarios as an end user, we only consider adversaries with such a restricted capability.
As we briefly mentioned earlier, this setting renders existing attack methods inapplicable, as they require query capabilities beyond decision~\cite{chen2021real, kim2024scores, galbally2010vulnerability, sharif2016accessorize} or prior knowledge of the enrolled identity~\cite{damer2018morgan, kim2025non, dong2019efficient}.
In such a strong adversarial capability, the adversary is often assumed as the \textit{insider} of the system, such as a service developer, who can silently query the backbone template extractor or API to obtain decision scores.
Furthermore, several reconstruction attacks require the knowledge of the template, assuming that the adversary can obtain an unprotected template through hacking or a data breach~\cite{mai2018reconstruction, shahreza2024vulnerability, otroshi2023face}.
However, insider access to score outputs or direct exposure of face templates is often considered an unrealistic capability according to the operational scenario assumptions documented in NIST and ISO/IEC standards~\cite{grother2013biometric, iso2022iso24745, iso2025iso19792}.

In summary, these representative lines of work fundamentally rely on capabilities beyond the end-user adversary, except for the baseline \textit{zero-effort imposter} based on FMR~\cite{grother2013biometric}.
This motivates the need for new techniques under this constrained yet realistic setting.

\section{Our Attack}\label{sec:OurAttack}

We now present our impersonation attack.
We first briefly introduce MasterFace-based impersonation attacks and our viewpoint on understanding them as the problem of finding a \net, i.e., maximum coverage over the template space.
We then show that the legitimate API queries allow the adversary to build a \net that reflects the template distribution of the target FRS, whose corresponding faces can further be recovered by score-based reconstruction attacks.

\subsubsection*{\bf MasterFaces: A Brief Overview}
MasterFaces~\cite{nguyen2020generating, shmelkin2021generating, terhorst2022limited, nguyen2022master, friedlander2022generating}, a special case of the wolf attack~\cite{une2007wolf, otsuka2013wolf} for faces, aim to find a face that could be identified as many other identities as possible.
More formally, for a database $\mathtt{DB}$ of facial images, the Masterface for a template extractor $F$ with a threshold $\tau$ is defined as follows:
\begin{align}\label{eq:MF}
    \mathfrak{b}_{\mathrm{MF}} := {\arg\max}_{ \mathfrak{b} \in \mathcal{B}} \sum_{\mathfrak{v} \in \mathtt{DB}} \mathds{1}(d_{\mathcal{X}}(F(\mathfrak{b}),  F(\mathfrak{v})) < \tau).
\end{align}
After finding the MasterFaces from Eq.~\eqref{eq:MF}, the adversary conducts an impersonation attack by querying the MasterFace to the target FRS; if multiple attempts are allowed, then the adversary finds other MasterFaces, excluding the previously tried ones.

Despite its clear formulation, MasterFaces from prior studies exhibit several limitations, rendering impersonation attacks based on them no better than the FMR-based baseline.
First, to solve Eq.~\eqref{eq:MF}, all prior studies~\cite{nguyen2020generating, shmelkin2021generating, terhorst2022limited} relied on evolutionary algorithms over the latent space of the generative models, e.g., the latent space evolution algorithm~\cite{hansen2003reducing, bojanowski2017optimizing} with StyleGAN~\cite{karras2019style, karras2020analyzing}.
Furthermore, several studies have shown that the MasterFaces are too sensitive to the choice of $F$~\cite{terhorst2022limited, nguyen2022master, friedlander2022generating}, hence they are barely transferable, i.e., MasterFaces from the adversary's own template extractor $F$ are no longer effective to the unknown template extractor of the target FRS.
More importantly, Terh\"ost et al.~\cite{terhorst2022limited} pointed out that several works reported results on outdated FRSs even at that time; finding MasterFaces from the prior heuristic-based algorithms was no longer effective for modern state-of-the-art FRSs~\cite{boutros2022elasticface, meng2021magface, terhorst2023qmagface}.
For these reasons, impersonation attacks based on prior MasterFaces do not pose a serious security threat against modern FRSs.

\subsection{MasterFace Impersonation Attack via \net}\label{sec:net}

Let us denote $\mathfrak{b}$ the target identity's facial image enrolled into the target FRS.
Instead of directly searching over $\mathcal{B}$, we turn our attention to the template space $(\mathcal{X}, d_{\mathcal{X}})$, and consider the following problem: finding $Q$ templates where at least one of them is sufficiently close to the template $t$ of $\mathfrak{b}$.
To simplify the problem, likewise to previous MasterFaces, we assume that the distribution $\mathcal{D}$ of templates is publicly known, and the adversary knows the template extractor $F$.
We formalize the above problem as follows:
\begin{definition}\label{def:searchVrfy}
    For the template extractor $F$ of $\mathsf{FRS}$, let $\mathcal{D}$ be the distribution defined over $\mathcal{X}$ of templates extracted from $F$. We consider the following search problem for $t \leftarrow \mathcal{D}$:
    \begin{itemize}
        \item\texttt{[Instance]} The template extractor $F$, the threshold $\tau > 0$, and an integer $Q \in \mathbb{N}$.

        \item\texttt{[Problem]} Find a set $\{z_{1}, \dots, z_{Q}\}$ satisfying $\exists i \in [Q]$ such that $d_{\mathcal{X}}(z_{i}, t) < \tau$.
    \end{itemize}
\end{definition}
Note that solving the above problem per se does not imply the success of the impersonation attack in our threat model because (i) finding such a template does not imply obtaining the corresponding faces, and (ii) the template extractor $F$ is unknown to the adversary.
At this moment, we put these issues aside and focus on solving the simplified problem; we will address them in later sections.

\subsubsection*{\bf Finding \net from Maximum Likelihood Strategy}
To solve the problem in Definition~\ref{def:searchVrfy}, our key strategy is to maximize the likelihood that the target template $t$ belongs to at least one of the $\tau$-neighborhoods of $z_{i}$ over $\mathcal{X}$.
That is, we aim to find the best plausible covering $\cup_{i=1}^{Q} \mathbb{B}_{\mathcal{X}}(z_{i}; \tau)$ to capture $t$, where $\mathbb{B}_{\mathcal{X}}(z;\tau):= \{ w \in \mathcal{X}: d_{\mathcal{X}}(w,z) < \tau  \}$ for $z \in \mathcal{X}$.
We call the centers $\{z_{1}, \dots, z_{Q}\}$ of this covering \net, and finding such a \net can be written as the following optimization problem:
\begin{align}\label{eq:contopt}
    \text{Maximize}\quad  &\mathrm{Pr}[t \in \cup_{i=1}^{Q} \mathbb{B}_{\mathcal{X}}(z_{i}; \tau) \;|\; t \leftarrow \mathcal{D}] \\ 
    \text{s.t.}\quad &z_{1},\dots,z_{Q} \in \mathcal{X} \nonumber
\end{align}
This strategy implicitly leverages the fact that biometric templates are not uniformly distributed in practice, sharing the same insight as MasterFaces.
Hence, carefully chosen centers $z_{i}$ cover significantly more mass than na\"ive independent trials or just sampling non-overlapping balls. 

Directly solving Eq.~\eqref{eq:contopt} would be cumbersome because the distribution $\mathcal{D}$ cannot be expressed in a closed form.
To address this, we adopt a Monte-Carlo-style approximation to estimate the probability in Eq.~\eqref{eq:contopt}.
More precisely, for the set of biometrics $\mathtt{DB}$ of size $M$ sampled from $\mathcal{B}$, the adversary can regard the following quantity as the objective function in Eq.~\eqref{eq:contopt}:
\begin{align}
    \frac{1}{M}\Big| \{ \mathfrak{b} \in \mathtt{DB}: F(\mathfrak{b}) \in \cup _{i=1}^{Q}\mathbb{B}_{\mathcal{X}}(z_{i}; \tau) \} \Big |
\end{align}

That is, for the set $F(\mathtt{DB}):=\{F(\mathfrak{b}) : \mathfrak{b} \in \mathtt{DB}\}$, finding \net, i.e., $z_{1}, \dots z_{Q}$, in Eq.~\eqref{eq:contopt}, can be viewed as finding $Q$ elements of $\mathcal{X}$ whose $\tau$-neighborhoods maximally cover $F(\mathtt{DB})$.
This is a special instance of the maximum coverage problem (MCP)~\cite{nemhauser1978analysis}, a famous classical problem in combinatorial optimization.

\subsubsection*{\bf Solving Maximum Coverage Problem}

Exactly solving the MCP is known to be NP-hard.
Nevertheless, its polynomial-time approximations, such as greedy-based~\cite{nemhauser1978analysis}, have been proposed.
We design a hybrid algorithm that clusters the templates first and then applies the greedy algorithm.
In the first layer of the algorithm, the adversary runs a clustering algorithm, such as K-Means~\cite{lloyd1982least, dhillon2001concept}, obtaining $\kappa Q$ candidate clustering centers for some $\kappa > 1$.
Then, for some candidate threshold $\tau$, the adversary selects $Q$ final centers whose $\tau$-neighborhood balls contain the largest number of templates via greedy search.
Although the latter process would not give an optimal solution, the submodularity of the objective function of the MCP guarantees that this approach gives $(1-e^{-1})$-approximated solution in polynomial time~\cite{nemhauser1978analysis}.
We provide a complete description of the algorithm in Appendix~\ref{sec:MCPDetail}.

\subsubsection*{\bf Remark on the Threshold $\tau$}
At this moment, for the ease of explanation, we assume that the adversary uses the threshold parameter $\tau$ when constructing a \net.
In our threat model, however, $\tau$ is unknown to the adversary and will be treated as an additional tunable hyperparameter.

\subsection{Exploiting Public API Queries}\label{sec:APIquery}

We now show how to legitimately leverage the commercial API queries of the target FRS, addressing assumptions we admitted for the ease of explanation, i.e., (i) finding representative facial images corresponding to the found \net, and (ii) getting rid of the dependence of \net on $F$.
Note that the adversary can obtain the confidence scores from pairs of its own chosen faces only; this does not grant any access to the target FRS, its internal decision scores, or the enrolled identity. 

For the ease of explanation, we assume that commercial APIs return the cosine similarity, i.e., the inner product between (normalized) feature vectors, instead of the confidence score.
Several studies showed that such confidence scores can be approximated as a function of the cosine similarity, e.g., a logistic function~\cite{kim2024scores, kim2025non}.
We provide a detailed explanation in Section~\ref{sec:detail}.

\subsubsection*{\bf Gram Decomposition of Scores}
For a positive semi-definite matrix $P \in \mathbb{R}^{d \times d}$, Gram decomposition aims to find a matrix $L\in \mathbb{R}^{d \times d}$ such that $P = L^{T}L$.
Special cases of Gram decompositions are well-known, e.g., Cholesky decomposition or matrix square root derived from spectral decomposition.
In addition, if we obtain two decompositions $P = L^{T}L = M^{T}M$, then it is well-known that there always exists a unitary matrix $R\in \mathbb{R}^{d \times d}$ such that $L = RM$.

Keeping in mind these mathematical facts, our key observation is that the adversary can construct a score matrix $S \in \mathbb{R}^{d \times d}$ from score queries whose entries consist of the pairwise cosine similarity between the queried biometric samples.
Since the cosine similarity is equivalent to the inner product between two unit vectors, we can write $S = A^{T}A$ where $A \in \mathbb{R}^{d \times d}$ consists of the feature vector of queried biometric samples extracted from the API\footnote{Of course, the dimension of templates in the API is unknown to adversary. We will carefully handle this issue in Section~\ref{sec:detail}.}.
Hence, $S$ is positive semi-definite, and the adversary can apply the Gram decomposition on $S$ to obtain $S=L^{T}L$, which guarantees that $L = RA$ for some unitary matrix $R \in \mathbb{R}^{d \times d}$.
Although the adversary does not know both $A$ and $R$, since $R$ can be viewed as the \textit{isometry}, i.e., distance-preserving mapping, the adversary can learn the geometric structure of feature vectors up to a global isometry.

To be precise, let $\mathtt{DB}_{P} := \{ \mathfrak{b}_{1}, \dots, \mathfrak{b}_{d} \}$ be a set of facial images. 
The adversary can obtain scores via queries of pairs of facial images in $\mathtt{DB}_{P}$, namely, $\{ d_{\mathcal{X}}(F_{T}(\mathfrak{b}_{i}), F_{T}(\mathfrak{b}_{j})) \}_{i \neq j}$, where $F_{T}$ is the template extractor of the target API.
The adversary then constructs a score matrix $S = A^{T}A$ where $A$ is the matrix whose $i^{\text{th}}$ column consists of $F_{T}(\mathfrak{b}_{i})$.
Finally, the adversary computes the Gram decomposition on $S$, obtaining $L=RA$ for some unitary matrix $R$.

Once obtaining $L$, the adversary can obtain $R \cdot F_{T}(\mathfrak{b}^{\ast})$ for any $\mathfrak{v}$ with score queries for $\{d_{\mathcal{X}}(F_{T}(\mathfrak{v}), F_{T}(\mathfrak{b}_{i})) \}_{i=1}^{d}$.
More precisely, let $\mathtt{DB}_{N} := \{\mathfrak{v}_{i}\}_{i=1}^{D}$ be a set of biometrics the adversary wishes to obtain the features of the form $R \cdot F_{T}(\mathfrak{v}_{i})$. 
By exploiting score queries, the adversary can construct a matrix $S_{B} = B^{T}A \in \mathbb{R}^{D \times d}$, where $B \in \mathbb{R}^{d \times D}$ is a matrix whose $i^{\text{th}}$ column vector is $F_{T}(\mathfrak{b}_{i}^{\ast})$.
Along with the fact that $I = R^{T}R$, we can observe that 
\begin{align}\label{eq:GramTechnique}
    S_{B} \cdot L^{-1} = B^{T}(R^{T}R)A \cdot (RA)^{-1} = (RB)^{T}
\end{align}
That is, $RB = (S_{B} \cdot L^{-1})^{T} \in \mathbb{R}^{d \times D}$, provided that $L$ is full-rank\footnote{We will show how to handle this condition in Section~\ref{sec:detail}.}.

To sum up, for some unknown unitary matrix $R$, the adversary can obtain the \textit{rotated} feature vectors by $R$ extracted from the API.
Since all these obtained features share the same isometry $R$, the adversary can leverage these rotated features to extract distributional traits of the target API's template extractor.

\subsubsection*{\bf Isometry Invariance of \net}
With the rotated features by $R$, the adversary can still exploit our \net-based attack strategy because $R$ preserves the geometric structure up to a global isometry in $\mathcal{X} = \mathbb{S}^{d-1}$.
In particular, MCP instances are invariant under isometry. 
More precisely, let $\phi: \mathcal{X} \rightarrow \mathcal{X}$ be an isometry, i.e., $\phi(x) = Rx$. 
Then we can observe that 
\begin{align}
    F(\mathfrak{b}) \in \cup_{i=1}^{Q}\mathbb{B}_{\mathcal{X}}(z_{i}; \tau) \iff \phi (F(\mathfrak{b})) \in \cup_{i=1}^{Q}\mathbb{B}_{\mathcal{X}}(\phi (z_{i}); \tau) \nonumber
\end{align}
as $\phi$ preserves the distance relationship.
That is, if $z_{1}^{\ast}, \dots, z_{Q}^{\ast}$ is the solution over the set $\{F(\mathfrak{b}): \mathfrak{b} \in \mathcal{B}^{\ast}\}$, then $\phi(z_{1}^{\ast}), \dots, \phi(z_{Q}^\ast)$ is the solution over $\{\phi(F(\mathfrak{b})): \mathfrak{b} \in \mathcal{B}^{\ast}\}$.
Therefore, if the adversary finds a \net from rotated features by $R$, then it should be the rotated \net derived from the original features. 
To help understand, we illustrate the above observation in Fig.~\ref{fig:invariance}.

\begin{figure}[t]
    \centering
    \includegraphics[width=.95\linewidth]{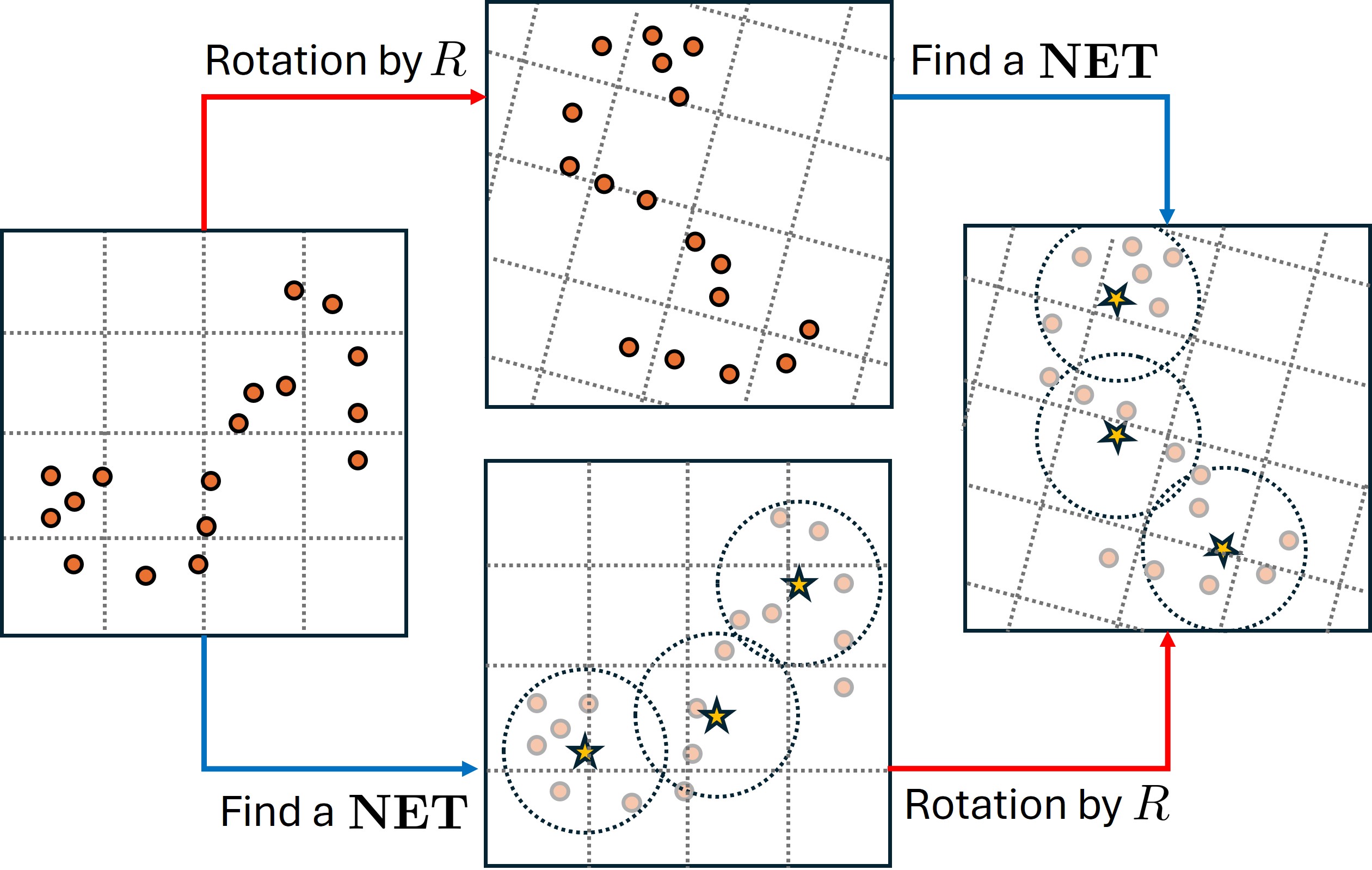}
    \vspace{-10pt}
    \caption{Isometry invariance of \net (Yellow stars) from samples (Orange dots) over the template space.}
    \label{fig:invariance}
    \vspace{-10pt}
\end{figure}

\subsubsection*{\bf Recovering Faces from \net}
Finally, the adversary obtains the rotated \net $\{R z_{1}^{\ast}, \dots, R  z_{Q}^{\ast} \}$ after solving the MCP.
Although the actual \net without rotation is still unknown, we observe that the adversary can obtain the cosine similarity between \net and the template $F_{T}(\mathfrak{v})$ of any arbitrary biometric samples $\mathfrak{v}$ through additional API queries, namely,
\begin{align}\label{eq:scoreOracle}
    \langle z_{i}^{\ast}, F_{T}(\mathfrak{v}) \rangle = \langle Rz_{i}^{\ast}, R \cdot F_{T}(\mathfrak{v}) \rangle = \langle Rz_{i}^{\ast}, (\! \underbrace{(F_{T}(\mathfrak{v})^{T} \cdot A)}_{\text{From API queries}}\!  \cdot L^{-1})^{T} \rangle
\end{align}
The last equality comes from Eq.~\eqref{eq:GramTechnique}.

This equality implies that the adversary can launch score-based reconstruction attacks~\cite{razzhigaev2021darker, vendrow2021realistic, jung2024face, kim2024scores}, which aim to reconstruct a facial image corresponding to the target identity through direct score queries to the target FRS.
Although these attacks per se are not applicable for impersonation in our threat model, in our attack pipeline, we can utilize them as a subroutine to recover representative biometric samples corresponding to \net by implicitly constructing a score oracle according to Eq.~\eqref{eq:scoreOracle}. 
As a result, both simplifying assumptions introduced for the ease of explanation, namely, access to the template extractor and direct instantiation of \net centers, are fully eliminated without granting any additional privilege over the target BAS or the enrolled identity.

\begin{figure}[t]
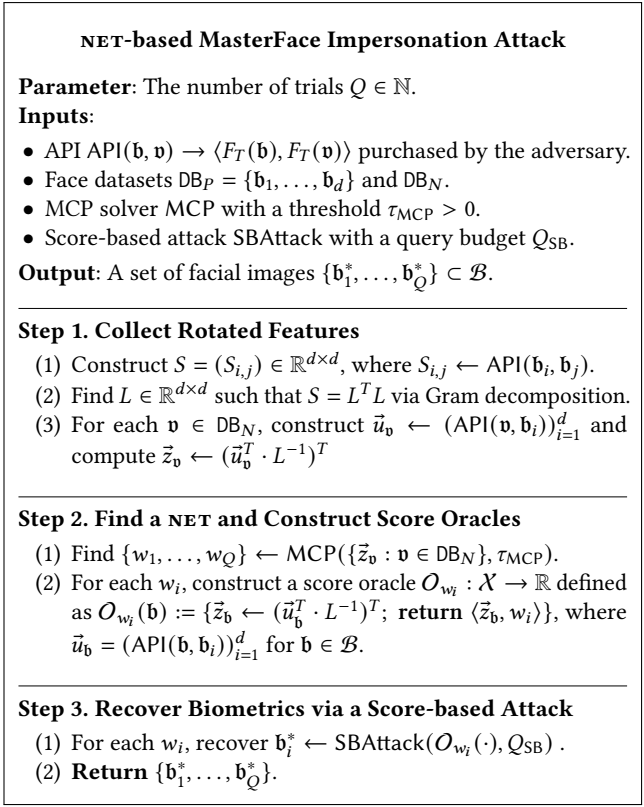

\fbox{
\parbox{.95\linewidth}{

\vspace{5pt}
\begin{center}
    % \textbf{Distribution-Aware Impersonation Attack}
    \textbf{\net-based MasterFace Impersonation Attack}
\end{center}
\vspace{5pt}

\textbf{Parameter}: The number of trials $Q \in \mathbb{N}$.

\textbf{Inputs}:
\begin{itemize}[leftmargin=10pt]
    \item API $\mathsf{API}(\mathfrak{b}, \mathfrak{v}) \rightarrow \langle F_{T}(\mathfrak{b}), F_{T}(\mathfrak{v}) \rangle$ purchased by the adversary.
    
    \item Face datasets $\mathtt{DB}_{P}=\{\mathfrak{b}_{1}, \dots, \mathfrak{b}_{d}\}$ and $\mathtt{DB}_{N}$.

    \item MCP solver $\mathsf{MCP}$ with a threshold $\tau_{\mathsf{MCP}} > 0$.

    \item Score-based attack $\mathsf{SBAttack}$ with a query budget $Q_{\mathsf{SB}}$.
\end{itemize}

\textbf{Output}: A set of facial images $\{\mathfrak{b}_{1}^{\ast}, \dots, \mathfrak{b}_{Q}^{\ast} \} \subset \mathcal{B}$.

\rule{\linewidth}{0.1mm}

\textbf{Step 1. Collect Rotated Features}
\begin{enumerate}[leftmargin = 20pt]
    
\item Construct $S=(S_{i,j}) \in \mathbb{R}^{d\times d}$, where $S_{i,j} \gets \mathsf{API}(\mathfrak{b}_{i}, \mathfrak{b}_{j})$.

\item Find $L \in \mathbb{R}^{d \times d}$ such that $S = L^{T}L$ via Gram decomposition.

\item For each $\mathfrak{v} \in \mathtt{DB}_{N}$, construct $\vec{u}_{\mathfrak{v}} \gets  (\mathsf{API}(\mathfrak{v}, \mathfrak{b}_{i}) )_{i=1}^{d}$ and compute $\vec{z}_{\mathfrak{v}} \gets (\vec{u}_{\mathfrak{v}}^{T} \cdot L^{-1})^{T}$

\end{enumerate}

\rule{\linewidth}{0.1mm}

\textbf{Step 2. Find a \net and Construct Score Oracles}
\begin{enumerate}[leftmargin = 20pt]
    
\item Find $\{w_{1}, \dots, w_{Q}\} \gets \mathsf{MCP}(\{\vec{z}_{\mathfrak{v}}: \mathfrak{v} \in \mathtt{DB}_{N} \}, \tau_{\mathsf{MCP}})$. 

\item For each $w_{i}$, construct a score oracle $\mathcal{O}_{w_{i}}: \mathcal{X} \rightarrow \mathbb{R}$ defined as $\mathcal{O}_{w_{i}}(\mathfrak{b}) := \{
\vec{z}_{\mathfrak{b}}\gets (\vec{u}_{\mathfrak{b}}^{T} \cdot L^{-1})^{T};\; \mathbf{return} \;\langle \vec{z}_{\mathfrak{b}}, w_{i} \rangle
\}$, where $\vec{u}_{\mathfrak{b}}=  (\mathsf{API}(\mathfrak{b}, \mathfrak{b}_{i}) )_{i=1}^{d}$ for $\mathfrak{b} \in \mathcal{B}$.

\end{enumerate}

\rule{\linewidth}{0.1mm}

\textbf{Step 3. Recover Biometrics via a Score-based Attack}
\begin{enumerate}[leftmargin = 20pt]

\item For each $w_{i}$, recover $\mathfrak{b}_{i}^{\ast} \gets \mathsf{SBAttack}(\mathcal{O}_{w_{i}}(\cdot), Q_{\mathsf{SB}})$ .

\item \textbf{Return} $\{ \mathfrak{b}_{1}^{\ast}, \dots, \mathfrak{b}_{Q}^{\ast} \}$.
\end{enumerate}
}
}
\vspace{-10pt}
\caption{Full Pipeline of the Proposed Attack}
\label{fig:pipeline}
% \Description[<short description>]{<long description>}

\vspace{-10pt}
\end{figure}

\subsection{Full Pipeline of the Attack}\label{sec:pipeline}
With the above-mentioned ideas, we present the full pipeline of the proposed impersonation attack based on \net.
Before starting the attack, the adversary first prepares public sets of biometrics $\mathtt{DB}_{P}$ and $\mathtt{DB}_{N}$.
The adversary can leverage public training datasets or their refinements, which would broadly reflect the population distribution.
With this dataset, the adversary first collects the rotated features via the Gram decomposition idea and API queries (Step 1), and then finds a \net by solving the MCP over the rotated features (Step 2).
Afterwards, the adversary now reconstructs the corresponding biometric samples through score-based attacks and submits them to the target FRS (Step 3).
We describe the full pipeline in Fig.~\ref{fig:pipeline}.
We will provide detailed explanations for each attack component in the next section.

\subsubsection*{\bf API Query Cost}

In our attack, the adversary makes API queries for (Step 1) and (Step 3).
For the former, constructing $S$ and $\{\vec{z}_{\mathfrak{v}}: \mathfrak{v} \in \mathtt{DB}_{N} \}$ require $\frac{d(d-1)}{2}$ and $|\mathtt{DB}_{N}| \cdot d$ queries, respectively.
On the other hand, the number of required queries in the latter depends on the choice of the score-based attack. 
As the adversary purchases the API as an end user, reducing API queries primarily affects the financial cost of the attack.

\section{Implementation Techniques and Extensions}

We present implementation-level details and extensions of the proposed \net-based impersonation attack.
First, we provide several details in our attack that we have postponed, including the rank deficiency of the score matrix $S$, handling unknown dimensionality of the templates, and the choice of the score-based attack.
In addition, we extend our attack to identification scenarios where many identities are enrolled in the system.

\subsection{Additional Techniques and Details}\label{sec:detail}

\subsubsection*{\bf From Confidence Scores to Cosine Similarity}
Many real-world API return confidence scores, and the adversary should convert them to cosine similarity values in order to conduct our attack.
More formally, for $x, y \in \mathcal{X}$ and an unknown confidence score function $g: \mathcal{X} \times \mathcal{X} \rightarrow [0,1]$, the adversary needs to find an invertible function $h: \mathbb{R} \rightarrow [0,1]$ such that $h(\langle x, y\rangle) = g( x, y )$.
To this end, recent studies~\cite{knoche2023explainable, kim2024scores, kim2025non} model $h$ as a logistic function, namely, $h(s) = \frac{L}{1 + e^{-k \cdot (s - d_{0})}} + b$ for learnable parameters $L, k, b, d_{0}$, and fit these parameters via the DOGBOX algorithm~\cite{voglis2004rectangular} with pairs of facial images from widely-used benchmark datasets, e.g., LFW~\cite{huang2008labeled, sengupta2016frontal} or CFP-FP.
As demonstrated by score-based attacks~\cite{kim2024scores, kim2025non}, this method was enough for launching their attack on commercial APIs, e.g., Amazon Rekognition~\cite{rekognition}.
We directly use the numbers provided in the open-source implementation of Kim et al.~\cite{kim2024scores}.
%\footnote{\url{https://github.com/Cryptology-Algorithm-Lab/Scores_Tell_Everything_about_Bob/blob/main/utils/AWS.py}}.

\subsubsection*{\bf Handling Rank Deficiency of $S$}
Recall that our attack is well-defined only when $S$ has a full rank.
However, in practice, $S$ is often rank-deficient because of highly correlated feature vectors among queried samples.
To compensate for this, we clamp the eigenvalue of the score matrix to be larger than $\epsilon=10^{-8}$.
In addition, to minimize the numerical error, we use the eigendecomposition of a symmetric matrix $S$; more precisely, we decompose $S = V\Sigma V^{T}$ and set $L = \Sigma^{1/2} V^{T}$.
We confirmed that the maximum error between $S$ and $L^{T}L$ is less than $10^{-7}$ in double precision.

\subsubsection*{\bf Unknown Dimensionality and Its Reduction}
In practice, the adversary does not know the true feature dimensionality of the API.
In particular, if the adversary uses a smaller dimensionality, say $k < d$, then the information loss would occur from the Gram decomposition, and the strict isometry invariance property of \net no longer holds\footnote{We do not consider $k\ge d$ because there will be no information loss.}.
Nevertheless, we observe that a similar argument in Section~\ref{sec:APIquery} still holds, discovering the relationship between the found \net and the facial images $\mathtt{DB}_{P}$ used for constructing $S$.

The key observation is that, if we consider the singular value decomposition of $A=U \Sigma V^{T} \in \mathbb{R}^{d \times k}$ for some $k < d$, where $U \in \mathbb{R}^{d \times k}$ and $\Sigma, V \in \mathbb{R}^{k \times k}$, then $A^{T}A \in \mathbb{R}^{k \times k}$ can be written as $(\Sigma V^{T})^{T}(\Sigma V^{T})$.
Note that $\Sigma V^{T} = U^{T}A$.
Hence, if we apply Gram decomposition on $A^{T}A$, obtaining $L\in \mathbb{R}^{k \times k}$ such that $A^{T}A = L^{T}L$, $L$ can be written as $RU^{T}A$ for some unitary matrix $R \in \mathbb{R}^{k \times k}$.
Such a decomposition satisfies the following properties\footnote{We provide the detailed proofs in Appendix~\ref{sec:Proofs}.}.
\begin{itemize}[leftmargin=*]
    \item For $z \in \mathbb{R}^{d}$, $RU^{T}z = (z^{T}A \cdot L^{-1})^{T}$. In particular, $z^{T}A = (RU^{T}z)^{T}L$.

    \item For $x, y\in \mathbb{R}^{d}$, $\langle RU^{T}x, \frac{RU^{T}y}{\|RU^{T}y\|_{2}} \rangle = \langle x, \frac{AA^{\dagger}y}{\|AA^{\dagger}y\|_{2}} \rangle$,
\end{itemize}
where $A^{\dagger}$ is the pseudoinverse of $A$. 

The former property tells us that the adversary can obtain feature vectors from the API transformed by $RU^{T}$, analogous to rotated features in Section~\ref{sec:APIquery}.
On the other hand, the latter property implies that, when finding a \net over the transformed set $\{RU^{T}\cdot F_{T}(\mathfrak{b}) : \mathfrak{b} \in \mathtt{DB}_{N}\}$, with a constraint that each clustered point $\{z_{1}, \dots, z_{Q}\} \subset \mathbb{R}^{k}$ has a norm of 1, then it maximizes
\begin{align}
    \left | \left\{ \mathfrak{b} \in \mathtt{DB}_{N}: 
     F_{T}(\mathfrak{b}) \in \bigcup_{i=1}^{Q} \mathbb{B}_{\mathcal{X}} \left( \frac{AA^{\dagger}y_{i}}{\|AA^{\dagger} y_{i} \|_{2}}; \tau \right)
    \right\} \right|; \quad y_{i} = UR^{T}z_{i}.\nonumber
\end{align}
That is, the found \net lies in the column space of $A$, corresponding to the subspace spanned by $\mathtt{DB}_{P}$.
In summary, even under dimensionality reduction, the adversary can still construct an effective \net aligned with the template space of the API's template extractor.

\subsubsection*{\bf Score-based Reconstruction Attack}
The adversary should recover precise face images corresponding to the found \net to reduce the error from the mismatch between templates and recovered faces. 
To this end, we utilize a score-based reconstruction attack by Kim et al.~\cite{kim2024scores}, which can be processed without adaptive score queries.
In particular, thanks to non-adaptiveness, their attack can be seamlessly incorporated into our attack without requiring additional API queries, as we can reuse rotated templates from previous queries.
Furthermore, we utilize a black-box optimization algorithm as a post-processing step, setting the reconstructed faces from Kim et al.'s attack as a starting point.
More precisely, we apply natural evolution strategy (NES)~\cite{wierstra2014natural}, a widely adopted method in black-box adversarial attacks~\cite{ilyas2018black, chen2021real}.
We provide detailed explanations of Kim et al.'s attack and the NES postprocessing in Appendix~\ref{sec:SBDetail}.

\subsection{Extension to Face Identification Systems}\label{sec:extensions}

Throughout the paper, we mainly considered the verification scenario, where only one identity is enrolled in the system.
In contrast, face identification scenarios are also widely adopted in practical application scenarios of FRS deployments, where several identities are enrolled in the system, and for the given query, the system determines whether the corresponding identity is already enrolled in the system or not, such as the access control scenarios~\cite{DHSBACS, starAlliance}.
The identification scenario for 1:N matching is also included in the NIST FRVT benchmark~\cite{NISTFRVT}. 
The standard evaluation metric analogous to FMR in this setting is false positive identification rate (FPIR)~\cite{iso2021iso19795, NISTFRVT}, which is the probability that the system incorrectly determines an unauthorized identity as one of the enrolled ones.

We can naturally extend our attack to analyze FRSs in these identification scenarios.
The core idea is that we can view the matching test with each enrolled identity as checking whether it is covered by \net or not.
Analogous to Definition~\ref{def:searchVrfy}, let us consider the following search problem for the identification scenario:
\begin{definition}\label{def:searchID}
    For the template extractor $F$ of $\mathsf{FRS}$, let $\mathcal{D}$ be the distribution defined over $\mathcal{X}$ of templates extracted from $F$. We consider the following search problem for i.i.d. samples $t_{1},\dots,t_{N} \leftarrow \mathcal{D}$:
    \begin{itemize}[leftmargin=*]
        \item\texttt{[Instance]} The template extractor $F$, the threshold $\tau > 0$, and integers $Q, N \in \mathbb{N}$.

        \item\texttt{[Problem]} Find a set $\{t_{1}, \dots, t_{Q}\}$ satisfying $\exists i \in [Q], j\in[N]$ such that $d_{\mathcal{X}}(z_{i}, t_{j}) < \tau$.
    \end{itemize}    
\end{definition}
If we denote $\rho := \mathrm{Pr}[t \in \cup_{i=1}^{Q} \mathbb{B}_{\mathcal{X}}(z_{i}; \tau) \;|\; t \leftarrow \mathcal{D}]$ for the found \net $\{z_{1}, \dots, z_{Q}\}$ from Section~\ref{sec:net}, by the pairwise independence of $t_{1}, \dots, t_{N}$, we can deduce the following equality:
\begin{align}\label{eq:ID}
    \mathrm{Pr}[\exists i\in[Q], j \in[N] \text{ s.t. } d_{\mathcal{X}}(z_{i},t_{j}) < \tau] =  1 - (1-\rho)^{N}.
\end{align}
This corresponds to the probability where at least one of the facial images corresponds to \net wrongly identified as the enrolled member to the target system.
Since $1-x \le e^{-x}$, the L.H.S. of the above Eq.~\eqref{eq:ID} is at least $1-e^{-N\rho}$. 
That is, once $\rho$ is obtained and the desired level of FPIR, say $\xi$, is given, then we can estimate the upper bound of the maximum number of identities handled by the system as $N \le \rho^{-1}\ln(\frac{1}{1-\xi})$.

\section{Experimental Results}

We present experimental results of the proposed attack.
All the source code is written in PyTorch and NumPy.
For reproducibility, we publicly release our source code\footnote{\url{https://zenodo.org/records/20765343}}.
Given the potential for misuse, we deliberately release only a minimal yet essential part of our attack pipeline, including the MCP solver, score-based attack pipeline, and the overall attack pipeline.
The released implementation is not sufficient for directly attacking real-world FRSs, and we did not attempt to conduct the attack against real-world deployments. 

\subsection{Experimental Settings}\label{sec:expsetting}
\subsubsection*{\bf Evaluation Criteria and Baseline Setup}
To reflect realistic operational parameters of FRS deployments, we consider the following FMR levels: $10^{-3}$, $10^{-4}$, and $10^{-5}$.
For each FRS, we determine the decision threshold for each FMR level as follows: given a dataset consisting of facial images from distinct identities, we compute the similarity score of every pair of images and select the threshold such that the proportion of pairs exceeding the threshold matches the target FMR.
This procedure follows standard evaluation practices, e.g., NIST FRVT benchmark~\cite{NISTFRVT}.
In addition, considering typical rate-limiting scenarios, we set the maximum number of authentication trials as $Q=5$, $10$, and $30$.

Throughout experiments, we evaluate the impersonation rate (IR) of the adversary by measuring the number of successful authentications via one of $Q$ facial images chosen by the adversary.
More precisely, if the adversary selects $Q$ facial images $\{\mathfrak{v}_{1}, \dots, \mathfrak{v}_{Q}\}$, for target face dataset $\{ \mathfrak{b}_{1}, \dots, \mathfrak{b}_{n} \}$ and the decision threshold $\tau$ of the target FRS $F_{T}$, we measure
\begin{align}
    \mathrm{IR}:= \frac{|\{ \mathfrak{b}_{i}: \exists j \text{ s.t. } \langle F_{T}(\mathfrak{b}_{i}), F_{T}(\mathfrak{v}_{j}) \rangle > \tau   \}|}{n} \nonumber.
\end{align}

As a baseline, we consider the zero-effort impostor, as prior studies have shown that no existing MasterFaces-based attacks showed better impersonation rate than such an adversary under our threat model~\cite{terhorst2022limited, nguyen2022master, friedlander2022generating}.
This adversary randomly samples $Q$ facial images from the public facial dataset $\mathtt{DB}$.
Note that, in this case, the upper bound of the impersonation rate can be roughly estimated by $Q \cdot \mathsf{FMR}(\tau)$, where $\mathsf{FMR}(\tau)$ is the FMR of the target FRS at the threshold $\tau$. 
Hence, we report both the exact IR and its estimation from the baseline adversary.

% \begin{table}[t]
%     \caption{Placeholder for the specification of target FRSs \red{To be filled}}
%     \vspace{-10pt}
%     \centering
%     \begin{tabular}{c|c}
%          &  \\
%          & 
%     \end{tabular}
%     \label{tab:targetFRS}
%     \vspace{-10pt}
% \end{table}

\begin{table}[t]
    \caption{Description of FRSs and the local inverse model.}
    \vspace{-10pt}
    \resizebox{.95\linewidth}{!}{
    \begin{tabular}{ccc|c}%ccc}
    \hline
    \multicolumn{3}{c|}{Open-source FRS / Inverse Model} & Train Dataset \\ \hline %& \multicolumn{3}{c}{TAR@FAR(\%)} \\ \hline 
    \multicolumn{1}{c|}{Notation} & \multicolumn{1}{c|}{Architecture} & Loss & Name \\ \hline \hline %& \multicolumn{1}{c|}{LFW} & \multicolumn{1}{c|}{CFP-FP} & AgeDB \\ \hline \hline
    \multicolumn{1}{c|}{$F_{1}$} & \multicolumn{1}{c|}{ResNet-100} & ArcFace~\cite{deng2019arcface} & Glint360k~\cite{an2021partial} \\ \hline % & \multicolumn{1}{c|}{99.70@0.00} & \multicolumn{1}{c|}{98.71@0.06} & 97.47@0.87\\ \hline
    \multicolumn{1}{c|}{$F_{1}^{-1}$} & \multicolumn{1}{c|}{NbNet-B} & Perceptual~\cite{mai2018reconstruction} & MS1MV3~\cite{DBLP:conf/iccvw/DengGZDLS19} \\ \hline % & \multicolumn{1}{c|}{N/A}  & \multicolumn{1}{c|}{N/A} & N/A \\ \hline
    \multicolumn{1}{c|}{$F_{2}$} & \multicolumn{1}{c|}{ResNet-100} & ElasticFace~\cite{boutros2022elasticface} & MS1MV2~\cite{deng2019arcface} \\ \hline % & \multicolumn{1}{c|}{N/A} & \multicolumn{1}{c|}{N/A} & N/A \\ \hline
    % \multicolumn{1}{c|}{$F_{2}^{-1}$} & \multicolumn{1}{c|}{NbNet-B} & Perceptual & MS1MV3~\cite{DBLP:conf/iccvw/DengGZDLS19} \\ \hline % & \multicolumn{1}{c|}{N/A} & \multicolumn{1}{c|}{N/A} & N/A \\ \hline
    \multicolumn{1}{c|}{$F_3$} & \multicolumn{1}{c|}{IResNet-101} & ArcFace~\cite{deng2019arcface} & WebFace4m~\cite{zhu2021webface260m} \\ \hline %& \multicolumn{1}{c|}{99.73@0.07} & \multicolumn{1}{c|}{98.74@0.20} & 97.33@1.17\\ \hline
    \multicolumn{1}{c|}{$F_4$} & \multicolumn{1}{c|}{Vit-KPRPE~\cite{kim2024keypoint}} & AdaFace~\cite{kim2022adaface} & WebFace12m~\cite{zhu2021webface260m} \\ \hline % &  \multicolumn{1}{c|}{99.67@0.00} & \multicolumn{1}{c|}{98.71@0.09} & 97.07@0.83\\ \hline
    \multicolumn{1}{c|}{$F_\mathsf{A}$} & \multicolumn{3}{c}{AWS CompareFaces API~\cite{rekognition}} \\ \hline
    % \multicolumn{1}{c|}{$F_\mathsf{T}$} & \multicolumn{6}{c}{Tencent CompareFace API} \\ \hline
    \end{tabular}
    }
    \vspace{-10pt}
%\vspace{-3mm}
\label{tab:targetFRS}
\end{table}

\subsubsection*{\bf Face Datasets and Target FRSs}
We select various open-source and commercial FRSs as targets.
For open-source models, we adopt recent sophisticated FRS, including ElasticFace~\cite{boutros2022elasticface} $(F_{2})$, ArcFace~\cite{deng2019arcface} $(F_{3})$, and AdaFace~\cite{kim2022adaface} $(F_{4})$.
Their implementations and pre-trained parameters are all publicly available; for ElasticFace, we refer to their GitHub repository\footnote{\url{https://github.com/fdbtrs/ElasticFace}}, and we utilize \texttt{insightface}~\cite{insightface} and \texttt{CVLFace}~\cite{CVLface} libraries for ArcFace and AdaFace, respectively.
On the other hand, for the commercial API, we select the CompareFaces API~\cite{rekognition}, a part of the Amazon Rekognition API that provides the confidence score between 0 and 100 for a pair of queried faces.
We note that, to instantiate Kim et al.'s attack~\cite{kim2024scores}, the adversary needs to select its own FRS $(F_{1})$, along with its inverse model ($F_{1}^{-1}$).
We select such an FRS from \texttt{insightface}, ensuring that at least two of the architecture, loss, and dataset are different from the target FRSs $F_{2}$--$F_{4}$.
We use the NbNet~\cite{mai2018reconstruction} as an inverse model, following Kim et al.~\cite{kim2024scores}.
We summarize these target FRSs in Table~\ref{tab:targetFRS}.

% \red{SH: Need to explain the choice of $\mathtt{DB}_{P}$ somewhere in this paragraph.}

When constructing a \net ($\mathtt{DB}_{N}$), we use a subset of the MS1MV3 dataset~\cite{deng2019retinaface} by sampling the first image over 93,431 identities in total.
For the facial images corresponding to those of the target enrolled identities, we select two datasets: LFW~\cite{huang2008labeled} and CASIA-WebFace~\cite{yi2014learning}.
More precisely, for LFW, we use the faces in their verification suite, excluding redundant images.
For CASIA-WebFace, we select the first image of each identity, hence 10,572 images in total.
For $\mathtt{DB}_{P}$, we use the facial images from the inverse model output of PCA components of the WebFace42M dataset~\cite{zhu2021webface260m}.
We provide the detailed specifications of target FRSs and face datasets, including the measured thresholds of target FRSs for each FMR level, in Appendix~\ref{sec:ExpDetail}.

\begin{table}[t]
\centering
\caption{IR (\%) results on $F_2$ at various FMR levels.  
For \textbf{Ours}, the second row in each cell, highlighted as \textbf{\color{blue}bold}, indicates the improvement factor over the baseline.}
\vspace{-10pt}
\resizebox{1.0\linewidth}{!}{
\begin{tabular}{c|c|ccc|ccc}
\hline
\multirow{3}{*}{$Q$} & \multirow{3}{*}{Method} & \multicolumn{6}{c}{Target Dataset} \\ \cline{3-8} 
 &  & \multicolumn{3}{c|}{LFW} & \multicolumn{3}{c}{CASIA-WebFace} \\ \cline{3-8} 
 & & $10^{-3}$ & $10^{-4}$& \multicolumn{1}{c|}{$10^{-5}$} & $10^{-3}$ & $10^{-4}$& $10^{-5}$ \\ \hline \hline
\multirow{4}{*}{$5$}
 & $Q \cdot \mathsf{FMR}$ 
 & 0.5 & 0.05 & 0.005 & 0.5 & 0.05 & 0.005 \\ \cline{2-8}

 & Baseline 
 & 0.522 & 0.054 & 0.008 & 0.514 & 0.055 & 0.010 \\ \cline{2-8}

 & Ours
 & \begin{tabular}{c} 2.743 \\ \textbf{\color{blue}$\times$5.25} \end{tabular}
 & \begin{tabular}{c} 0.441 \\ \textbf{\color{blue}$\times$8.17} \end{tabular}
 & \begin{tabular}{c} 0.076 \\ \textbf{\color{blue}$\times$9.50} \end{tabular}
 & \begin{tabular}{c} 2.936 \\ \textbf{\color{blue}$\times$5.71} \end{tabular}
 & \begin{tabular}{c} 0.445 \\ \textbf{\color{blue}$\times$8.09} \end{tabular}
 & \begin{tabular}{c} 0.040 \\ \textbf{\color{blue}$\times$4.00} \end{tabular}
 \\ \hline

\multirow{4}{*}{$10$}
 & $Q \cdot \mathsf{FMR}$ 
 & 1 & 0.1 & 0.01 & 1 & 0.1 & 0.01 \\ \cline{2-8}

 & Baseline 
 & 1.041 & 0.109 & 0.017 & 1.025 & 0.111 & 0.021 \\ \cline{2-8}

 & Ours
 & \begin{tabular}{c} 4.741 \\ \textbf{\color{blue}$\times$4.55} \end{tabular}
 & \begin{tabular}{c} 0.550 \\ \textbf{\color{blue}$\times$5.05} \end{tabular}
 & \begin{tabular}{c} 0.049 \\ \textbf{\color{blue}$\times$2.88} \end{tabular}
 & \begin{tabular}{c} 5.210 \\ \textbf{\color{blue}$\times$5.08} \end{tabular}
 & \begin{tabular}{c} 0.788 \\ \textbf{\color{blue}$\times$7.10} \end{tabular}
 & \begin{tabular}{c} 0.112 \\ \textbf{\color{blue}$\times$5.33} \end{tabular}
 \\ \hline

\multirow{4}{*}{$30$}
 & $Q \cdot \mathsf{FMR}$ 
 & 3 & 0.3 & 0.03 & 3 & 0.3 & 0.03 \\ \cline{2-8}

 & Baseline 
 & 3.090 & 0.325 & 0.050 & 3.042 & 0.331 & 0.062 \\ \cline{2-8}

 & Ours
 & \begin{tabular}{c} 12.13 \\ \textbf{\color{blue}$\times$3.92} \end{tabular}
 & \begin{tabular}{c} 1.685 \\ \textbf{\color{blue}$\times$5.18} \end{tabular}
 & \begin{tabular}{c} 0.211 \\ \textbf{\color{blue}$\times$4.22} \end{tabular}
 & \begin{tabular}{c} 12.52 \\ \textbf{\color{blue}$\times$4.12} \end{tabular}
 & \begin{tabular}{c} 2.197 \\ \textbf{\color{blue}$\times$6.64} \end{tabular}
 & \begin{tabular}{c} 0.285 \\ \textbf{\color{blue}$\times$4.60} \end{tabular}
 \\ \hline
\end{tabular}
}
\vspace{-10pt}
\label{tab:open}
\end{table}

\subsubsection*{\bf Instantiation of Our Attack}

As we mentioned in Section~\ref{sec:detail}, we utilize score-based attacks to recover faces corresponding to \net.
We utilize Kim et al.'s attack~\cite{kim2024scores} based on their official implementation\footnote{\url{https://github.com/Cryptology-Algorithm-Lab/Scores_Tell_Everything_about_Bob}}, and we manually implement a latent space search variant of NES.
We provide the detailed pseudocodes for each attack algorithm and detailed parameter settings in Appendix~\ref{sec:SBDetail}.

Recall that our MCP solver is comprised of a combination of the K-Means and greedy algorithms.
To implement the K-Means, we use the \texttt{faiss} library~\cite{douze2025faiss} with GPU acceleration\footnote{\url{https://github.com/facebookresearch/faiss}}, finding $\kappa Q$ candidates of \net for some multiplier $\kappa > 1$.
After that, we apply the greedy search to obtain $Q$ centers whose $\tau_{\mathsf{MCP}}$-neighborhoods maximally cover the given dataset.
Detailed pseudocode of the MCP solver is provided in Appendix~\ref{sec:MCPDetail}.

% \red{SH: Need to be revised}.

\subsection{Attack Evaluation Results}

We present the IR results according to the settings mentioned above.
For the hyperparameters, $\tau_{\mathsf{MCP}}$ and $\kappa$, we use the values that yield the best IR results, as described in the following paragraphs.

\subsubsection*{\bf Open-Source FRS}
We first report IR results from our attack compared with the baseline.
We use the Gram decomposition dimension $k=400$, and enable the NES postprocessing.
We empirically found that $\tau_{\mathsf{MCP}} = 0.2$ and $\kappa = 10$ shows the best results.
Due to space constraints, we only provide the IR results for $F_{2}$; the other open-source FRSs ($F_{3}$, $F_{4}$) show a similar trend as $F_{2}$, and we provide the full experimental results in Appendix~\ref{sec:ExpDetail}.

The results are provided in Table~\ref{tab:open}.
From this table, we can observe that ours consistently achieves higher IRs across all the settings we considered, compared to those from the baseline and estimated from the FMR.
More precisely, when FMR is $10^{-3}$, $10^{-4}$, and $10^{-5}$, ours achieve $3.92\times$--$5.71\times$, $5.05\times$--$8.17\times$, and $2.88\times$--$9.50\times$ improvements, respectively, compared to the baseline.
Notably, we can observe that the relative IR increases for relatively small query budgets $Q=5,10$ and strict FMR levels $10^{-4}, 10^{-5}$, while the baseline shows almost similar results as those estimated from $Q \cdot \mathsf{FMR}$.
These results indicate that our attack becomes far more effective under realistic yet restrictive operational settings.

\begin{table}[t]
\centering
\caption{IR (\%) results on $F_\mathsf{A}$ at various FMR levels.  
For \textbf{Ours}, the second row in each cell, highlighted as \textbf{\color{blue}bold}, indicates the improvement factor over the baseline.}
\vspace{-10pt}
\resizebox{1.0\linewidth}{!}{
\begin{tabular}{c|c|ccc|ccc}
\hline
\multirow{3}{*}{$Q$} & \multirow{3}{*}{Method} & \multicolumn{6}{c}{Target Dataset} \\ \cline{3-8} 
 &  & \multicolumn{3}{c|}{LFW} & \multicolumn{3}{c}{CASIA-WebFace} \\ \cline{3-8} 
 &  & $10^{-3}$ & $10^{-4}$ & \multicolumn{1}{c|}{$10^{-5}$} 
    & $10^{-3}$ & $10^{-4}$ & $10^{-5}$ \\ \hline \hline

\multirow{4}{*}{$5$}
 & $Q \cdot \mathsf{FMR}$ 
 & 0.5 & 0.05 & 0.005 & 0.5 & 0.05 & 0.005 \\ \cline{2-8}

 & Baseline 
 & 0.493 & 0.049 & 0.007 & 0.558 & 0.055 & 0.008 \\ \cline{2-8}

 & Ours
 & \begin{tabular}{c} 1.365 \\ {\textbf{\color{blue}{$\times$2.77}}} \end{tabular}
 & \begin{tabular}{c} 0.143 \\ {\textbf{\color{blue}{$\times$2.92}}} \end{tabular}
 & \begin{tabular}{c} 0.046 \\ {\textbf{\color{blue}{$\times$6.57}}} \end{tabular}
 & \begin{tabular}{c} 1.764 \\ {\textbf{\color{blue}{$\times$3.16}}} \end{tabular}
 & \begin{tabular}{c} 0.158 \\ {\textbf{\color{blue}{$\times$2.87}}} \end{tabular}
 & \begin{tabular}{c} 0.014 \\ {\textbf{\color{blue}{$\times$1.75}}} \end{tabular}
 \\ \hline

\multirow{4}{*}{$10$}
 & $Q \cdot \mathsf{FMR}$ 
 & 1 & 0.1 & 0.01 & 1 & 0.1 & 0.01 \\ \cline{2-8}

 & Baseline 
 & 0.984 & 0.098 & 0.014 & 1.111 & 0.111 & 0.015 \\ \cline{2-8}

 & Ours
 & \begin{tabular}{c} 2.219 \\ {\textbf{\color{blue}{$\times$2.26}}} \end{tabular}
 & \begin{tabular}{c} 0.199 \\ {\textbf{\color{blue}{$\times$2.03}}} \end{tabular}
 & \begin{tabular}{c} 0.059 \\ {\textbf{\color{blue}{$\times$4.21}}} \end{tabular}
 & \begin{tabular}{c} 3.097 \\ {\textbf{\color{blue}{$\times$2.79}}} \end{tabular}
 & \begin{tabular}{c} 0.328 \\ {\textbf{\color{blue}{$\times$2.95}}} \end{tabular}
 & \begin{tabular}{c} 0.051 \\ {\textbf{\color{blue}{$\times$3.40}}} \end{tabular}
 \\ \hline

\multirow{4}{*}{$30$}
 & $Q \cdot \mathsf{FMR}$ 
 & 3 & 0.3 & 0.03 & 3 & 0.3 & 0.03 \\ \cline{2-8}

 & Baseline 
 & 2.920 & 0.294 & 0.041 & 3.290 & 0.332 & 0.045 \\ \cline{2-8}

 & Ours
 & \begin{tabular}{c} 5.218 \\ {\textbf{\color{blue}{$\times$1.79}}} \end{tabular}
 & \begin{tabular}{c} 0.817 \\ {\textbf{\color{blue}{$\times$2.78}}} \end{tabular}
 & \begin{tabular}{c} 0.114 \\ {\textbf{\color{blue}{$\times$2.78}}} \end{tabular}
 & \begin{tabular}{c} 7.441 \\ {\textbf{\color{blue}{$\times$2.26}}} \end{tabular}
 & \begin{tabular}{c} 1.039 \\ {\textbf{\color{blue}{$\times$3.13}}} \end{tabular}
 & \begin{tabular}{c} 0.153 \\ {\textbf{\color{blue}{$\times$3.40}}} \end{tabular}
 \\ \hline
\end{tabular}
}
\vspace{-10pt}
\label{tab:comm}
\end{table}

\subsubsection*{\bf AWS CompareFace API}
Before reporting the IR results, we briefly summarize the pricing policy of AWS Recognition API\footnote{\url{https://aws.amazon.com/rekognition/pricing/}}.
For each query or image enrollment, AWS CompareFace charges \$0.001 per API call.
As described in Section~\ref{sec:pipeline}, the adversary needs to query $\frac{k(k-1)}{2} + k \cdot |\mathtt{DB}_{N}|$ pairs of facial images to construct a \net, where $k$ is the dimension of the rotated templates obtained from the Gram decomposition.
AWS CompareFace supports batched score queries up to 4,096 images at the same cost as a single query.
This allows reducing the total query cost to $(k-1) + k \cdot \frac{|\mathtt{DB}_{N}|}{4096}$, in addition to enrolling $k+|\mathtt{DB}_{N}|$ facial images.

Based on the pricing policy above, we select $k=200$ and turn off the NES postprocessing, although NES could improve the attack performance.
Under this setting, the total API cost for obtaining all rotated templates is approximately \$100.
We do not apply NES post-processing due to the budget constraint; enabling this step would like further improve the IR.
Importantly, the cost of \$100 is incurred only once; after the rotated features are obtained, the adversary can construct the \net locally using the MCP solver for arbitrary authentication trial budgets, without further API queries.

We conduct our attack with hyperparameters $\tau_{\mathsf{MCP}}=0.3$ and $\kappa=10$, and the results are provided in Table~\ref{tab:comm}. 
Consistent with the behavior of open-source FRSs, the IRs from our attack consistently exceed those from the baseline, and the gap increases for small authentication trial budgets and strict FMR levels.
We observe that the IR gain from attacking AWS CompareFace is smaller than that observed for open-source FRSs, achieving $1.75\times$--$6.57\times$ across the experimental settings.
We attribute this phenomenon to the approximation error introduced by converting confidence scores to cosine similarities; the exact relationship between them is unknown in public.
Nevertheless, the adversary may further improve the FMR via additional API queries and NES-based postprocessing.

\subsubsection*{\bf Implication on Identification Scenarios}
As discussed in Section~\ref{sec:extensions}, our attack can be naturally extended to identification scenarios.
For the IR value $\rho$ at a verification threshold $\tau$, the FPIR among $N$ enrolled identities can be estimated by $1-(1-\rho)^{N} \approx 1-e^{-\rho N}$.
This relationship implies that the maximum number of identities safely handled by the system is inversely proportional to $\rho$.
Accordingly, the results in Table~\ref{tab:open} and~\ref{tab:comm} can be interpreted as the $\rho$ values at different thresholds and FMR levels, which directly determine the system's vulnerability in identification deployments.

\begin{table}[t]
\centering
\caption{IRs (\%) against FRSs $F_{1}$--$F_{4}$ and $F_{\mathsf{A}}$ from the adversary exploiting MasterFaces crafted from $F_{1}$.}
\vspace{-10pt}
\resizebox{.9\linewidth}{!}{
\begin{tabular}{c|c|ccc|ccc}
\hline
\multirow{3}{*}{$Q$} & \multirow{3}{*}{$F_T$} & \multicolumn{6}{c}{Target Dataset} \\ \cline{3-8}
 &  & \multicolumn{3}{c|}{LFW} & \multicolumn{3}{c}{CASIA-WebFace} \\ \cline{3-8}
 & & $10^{-3}$ & $10^{-4}$ & $10^{-5}$ & $10^{-3}$ & $10^{-4}$ & $10^{-5}$ \\ \hline \hline

\multirow{5}{*}{$5$}
 & $F_1$ & 1.929 & 0.322 & 0.013 & 2.766 & 0.369 & 0.038 \\ \cline{2-8}
 & $F_2$ & 1.042 & 0.090 & 0     & 1.506 & 0.171 & 0.019 \\ \cline{2-8}
 & $F_3$ & 1.351 & 0.180 & 0     & 1.478 & 0.256 & 0.038 \\ \cline{2-8}
 & $F_4$ & 0.875 & 0.090 & 0.026 & 1.288 & 0.180 & 0.019 \\ \cline{2-8}
 & $F_{\mathsf{A}}$ & 0.429 & 0.039 & 0     & 1.097 & 0.189 & 0.009 \\ \hline

\multirow{5}{*}{$10$}
 & $F_1$ & 5.055  & 0.733 & 0.154 & 6.924 & 1.250 & 0.218 \\ \cline{2-8}
 & $F_2$ & 3.344  & 0.347 & 0.013 & 3.628  & 0.673 & 0.161 \\ \cline{2-8}
 & $F_3$ & 4.270  & 0.733 & 0.090 & 5.001  & 1.099 & 0.123 \\ \cline{2-8}
 & $F_4$ & 2.662  & 0.270 & 0.051 & 3.561  & 0.616 & 0.066 \\ \cline{2-8}
 & $F_{\mathsf{A}}$ & 1.247 & 0.091 & 0.026 & 1.750 & 0.265 & 0.047 \\ \hline

\multirow{5}{*}{$30$}
 & $F_1$ & 11.81  & 2.071 & 0.257 & 15.34  & 3.021 & 0.417 \\ \cline{2-8}
 & $F_2$ & 7.897  & 1.183 & 0.206 & 9.945  & 1.639 & 0.284 \\ \cline{2-8}
 & $F_3$ & 9.029  & 1.415 & 0.180 & 11.78 & 1.856 & 0.322 \\ \cline{2-8}
 & $F_4$ & 7.460  & 1.080 & 0.193 & 8.591  & 1.658 & 0.189 \\ \cline{2-8}
 & $F_{\mathsf{A}}$ & 3.685 & 0.468 & 0.065 & 5.032 & 0.568 & 0.057 \\ \hline

\end{tabular}
}
\vspace{-10pt}
\label{tab:white}
\end{table}

\subsection{Ablation Study}

We present ablation studies on the core components of our attack, including the comparison of IRs from MasterFaces generated from a white-box surrogate model and the Gram decomposition dimension according to our analysis in Section~\ref{sec:detail}.

\subsubsection*{\bf Effect of Leveraging API Queries}
To quantify the IR gains from the exposition of API queries, we measure the IRs from the adversary who directly crafts MasterFaces on its surrogate model and uses them for conducting an impersonation attack against the target FRSs.
We set $F_{1}$ as the surrogate model, attacking the remaining FRSs as targets.
In addition, we also report the IR for attacking $F_{1}$ itself, which corresponds to the white-box attack scenario, i.e., the adversary has full knowledge of the target FRS's template extractor.

The results are provided in Table~\ref{tab:white}.
When conducting the attack against $F_{2}$--$F_{4}$ and $F_{\mathsf{A}}$, the MasterFaces from $F_{1}$ show IR values in the middle between those of the baseline and our attack.
On the other hand, the IR values from attacking the same $F_{1}$ serve as the upper bound of the IRs achieved by our attack, as this setting corresponds to the white-box attack.
This indicates that the Gram decomposition technique indeed contributes to crafting API-tailored \net.

\begin{figure}[t]
    \centering
    \includegraphics[width=.95\linewidth]{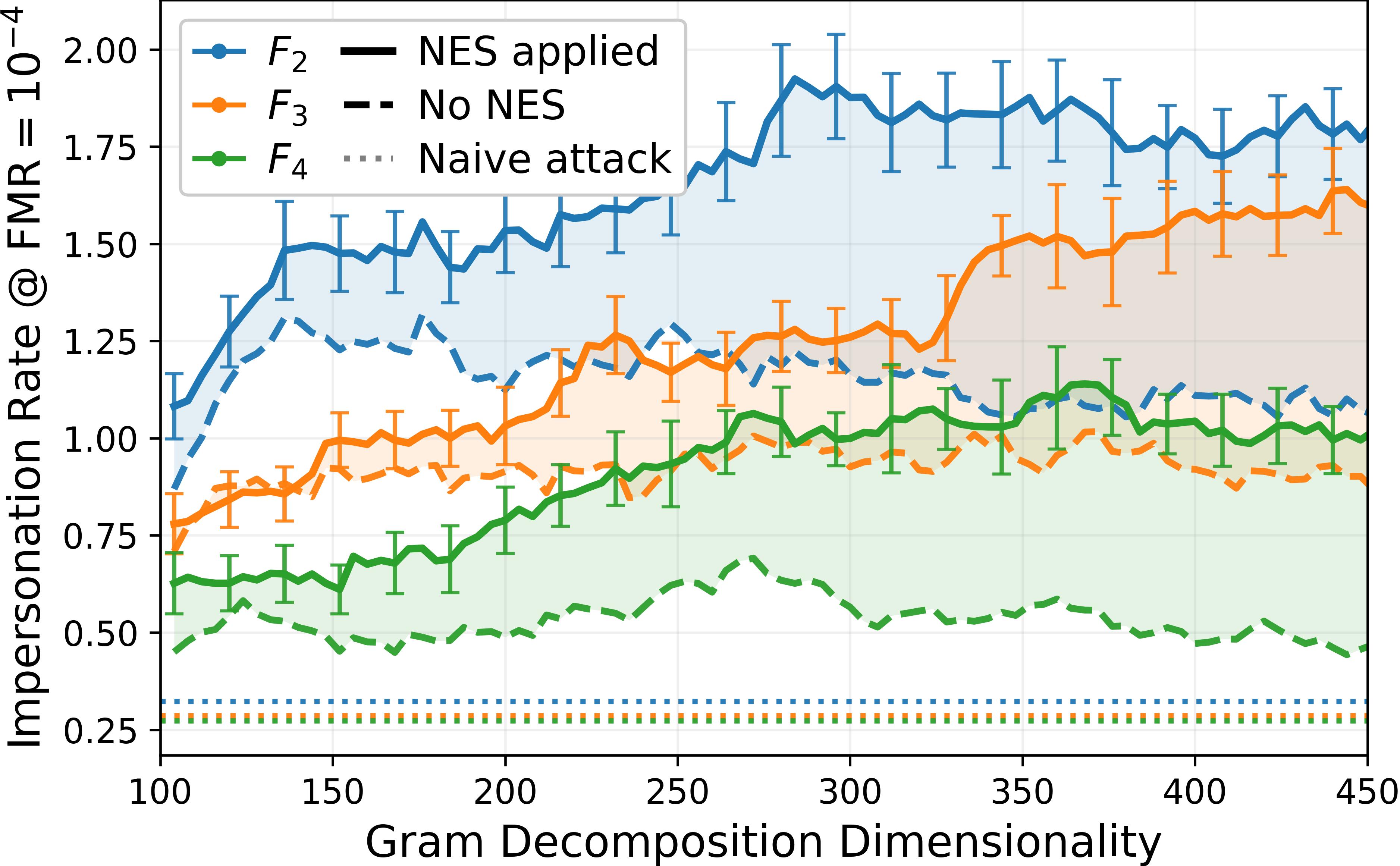}
    \vspace{-10pt}
    \caption{IRs (\%) with various Gram decomposition dimensionality. The colored area depicts the difference in IRs (\%) incurred by the NES postprocessing, and the error bars indicate the standard deviation.}
    \label{fig:NES_1}
    \vspace{-12pt}
\end{figure}

\subsubsection*{\bf Effect of Gram Decomposition Dimensionality}
As discussed in Section~\ref{sec:detail}, the dimensionality for the Gram decomposition determines the solution space of the \net in analogy with PCA-based dimensionality reduction.
Increasing the dimension enlarges the solution space where the \net can lie, thereby resulting in more coverage over the template space by the constructed \net.
To examine this, we sweep the decomposition dimension from 100 to 450 and evaluate the IR values from our attack with the NES postprocessing.
We fixed $\tau_{\mathrm{MCP}}=0.2$ and $\kappa=10$, and set $Q=5$ and the FMR as $10^{-4}$.
Since the NES is disabled when attacking $F_{\mathsf{A}}$, we report the results from $F_{1}$--$F_{4}$ only.
To account for the randomness of the NES, we report the mean IR over 30 independent runs, with error bars indicating the standard deviation.

The results are visualized in Fig.~\ref{fig:NES_1}.
As the dimension grows, we can observe that the IRs from our attack with NES postprocessing improve while those from turning off NES remain largely unchanged.
This behavior arises because Kim et al.'s reconstruction attack recovers faces from templates that lie in some low-dimensional subspace (see Appendix~\ref{sec:SBDetail} for a detailed explanation).
Consequently, although increasing the Gram decomposition dimension broadens the solution space for \net, the reconstructed faces themselves remain confined to the same subspace.
NES postprocessing bridges this mismatch by adapting the reconstructed faces to the expanded solution space.

\subsection{What Do Our MasterFaces Look Like?}

We discuss MasterFaces crafted from our attack, providing their visualizations and the faces they cover.

\subsubsection*{\bf Visualization of MasterFaces}
In our attack, the visual quality of the MasterFaces depends on the score-based reconstruction attack subroutine.
The NbNet used by Kim et al.'s attack~\cite{kim2024scores} reconstructs faces of $128\times 128$ resolution.
Since their attack can be combined with other inverse models, we also provide the visualization results from replacing the NbNet with Arc2Face~\cite{papantoniou2024arc2face} that produces $512\times 512$ resolution face images.
Fig.~\ref{fig:masterface} shows MasterFaces obtained from attacking AWS CompareFace API ($F_{\mathsf{A}}$).

\begin{figure}[t]
    \centering
    \includegraphics[width=.95\linewidth]{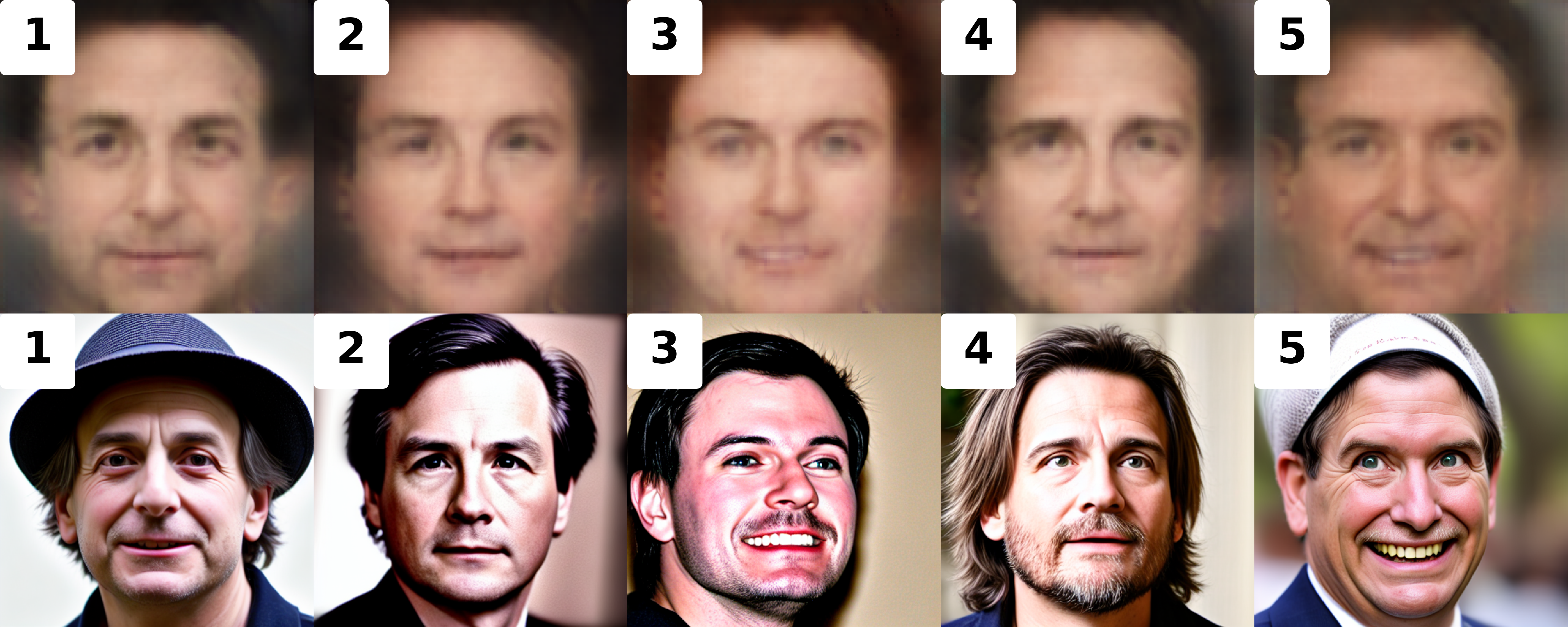}
    \vspace{-10pt}
    \caption{MasterFaces from our attack against $F_\mathsf{A}$ with NbNet (top) and Arc2Face (bottom). We set $Q=5$.}
    \label{fig:masterface}
    \vspace{-10pt}
\end{figure}

\subsubsection*{\bf Which Faces are Covered by MasterFaces?} 
Surprisingly, although the reconstructed MasterFaces in Fig.~\ref{fig:masterface} appear as faces of white males in the pixel space, they successfully cover target faces across diverse demographic groups.
In Fig.~\ref{fig:matched}, we visualize the covered facial images by MasterFaces in Fig.~\ref{fig:masterface} on $F_{\mathsf{A}}$ when the FMR is set to $10^{-4}$.
From this figure, we observe that MasterFaces, despite their white male appearance, cover faces from a wide range of demographic groups, including females and non-white individuals.
This observation is seemingly counterintuitive, but it highlights that the effectiveness of MasterFaces is governed by their coverage over the template space rather than the pixel space.
This behavior is consistent with the \net-based formulation of our attack, which explicitly targets non-uniformity in the template space.

In fact, this behavior is a direct consequence of the intrinsic properties of Kim et al.'s attack.
Recent studies~\cite{leroy2025attributes, kim2025non} have demonstrated that facial images of the same attributes, e.g., race or gender, form a subspace over the template space.
Kim et al.'s attack finds a projection of the target template onto the subspace spanned by the queried faces, which is $\mathtt{DB}_{P}$ in our attack instantiation.
The selected facial images to launch Kim et al.'s attack correspond to the PCA components of WebFace42M~\cite{zhu2021webface260m}, a large-scale web-crawled face dataset.
More generally, web-scraped face datasets have been widely reported to exhibit demographic imbalances, often favoring white and male subjects~\cite{wang2019racial, grother2019face}.
Hence, the recovered faces from Kim et al.'s attack are likely to be the projection of the ground truth \net onto the subspace correlated to the faces of white males.

\begin{figure}[t]
    \centering
    \includegraphics[width=.95\linewidth]{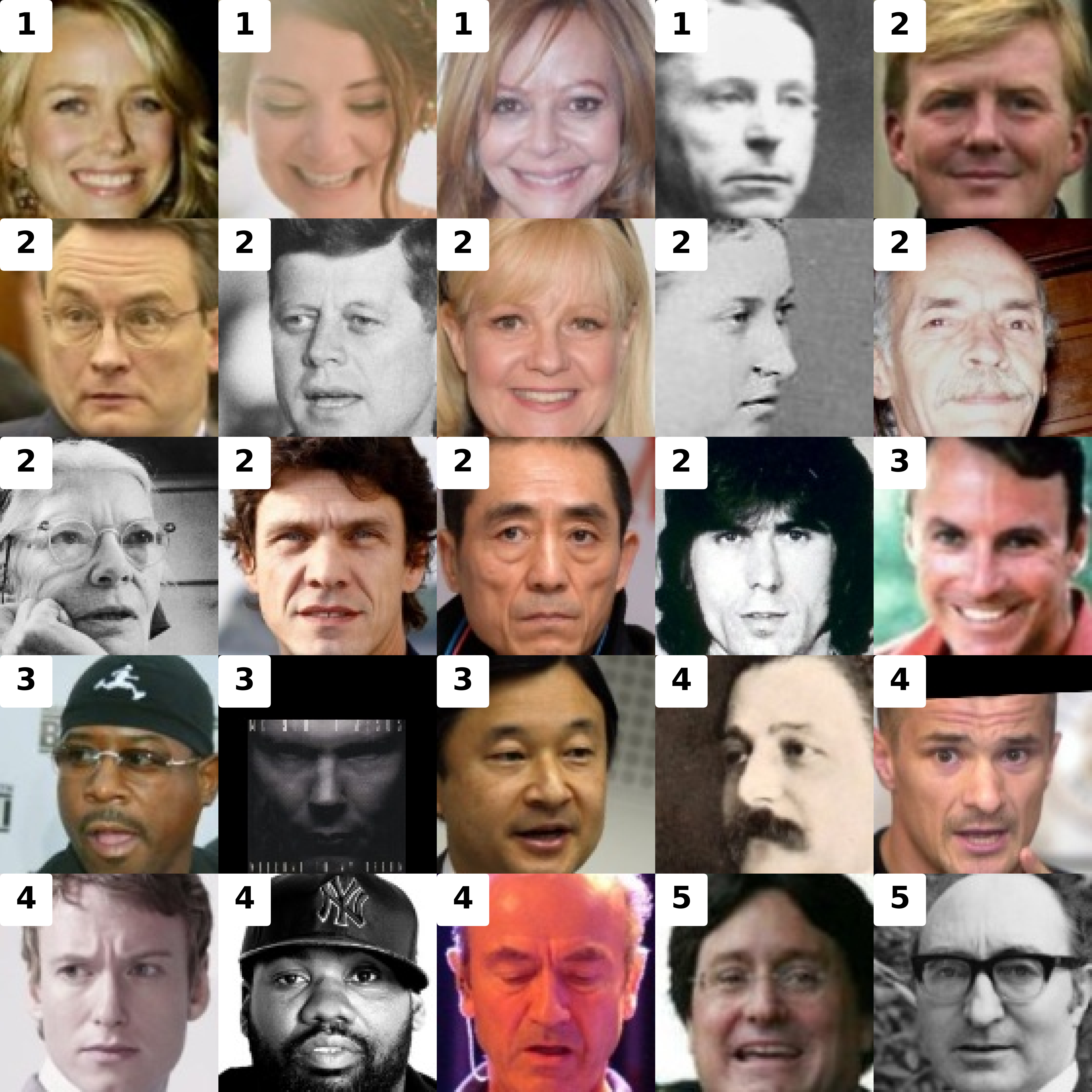}
    \vspace{-10pt}
    \caption{Faces covered by MasterFace against AWS at an FMR of $10^{-4}$ in Fig.~\ref{fig:masterface}, sampled from LFW and CASIA-WebFace. Each number represents the corresponding MasterFace.}
    \label{fig:matched}
    \vspace{-10pt}
\end{figure}

\section{Discussion}

\subsubsection*{\bf Potential Countermeasure}
We first recall that our attack consists of two core ingredients: \net and Gram decomposition.
The former enables finding the maximum coverage of the known template distribution, while the latter provides such a distribution from the API's template extractor through queries.
From this viewpoint, one possible countermeasure to our attack is to interrupt the latter phase, e.g., add random noise to the query response.
This significantly hinders the adversary's capability to construct an accurate \net, which is sensitive to even small distribution shifts.
This effect is further supported by our experimental observations: when attacking commercial APIs, the error derived from converting the confidence score to the cosine similarity makes it challenging to construct a \net.
We expect that differential privacy techniques on hyperspheres, e.g., directional differential privacy mechanism~\cite{weggenmann2021differential}, may provide a principled framework for formalizing and quantifying the effectiveness of such defenses in terms of the resulting quality of the \net.
We leave such an investigation as future work.

\subsubsection*{\bf Physical Realization of Our Attack}
When attacking real-world FRSs, the adversary may need to physically present facial images to the camera, e.g., 2D printing or 3D masks~\cite{sharif2016accessorize, liu2025projattacker, an2023imu}.
Nevertheless, we evaluated the digital MasterFaces produced from the proposed attack, as our main focus was to investigate and exploit the non-uniformity of face templates through API queries.
This is also the reason why we report FMR, which characterizes matcher-level behavior, instead of the false acceptance rate that incorporates system-level failures.
Specifically, the above-mentioned realization strongly depends on how to craft physical faces, which lies in the realm of face spoofing, and thus is orthogonal to our approach.
We leave assessing the vulnerability of real-world systems against our attack with spoofing as an important future direction.

\subsubsection*{\bf Applicability to Other Biometric Modalities}
Along with faces, various biometric traits have been widely used for authentication, e.g., fingerprints, voices, and irises, with various template representations.
Since rotation is still an isometry over the Euclidean space with $\ell_{2}$ distance, our pipeline is still applicable when the distance between templates is calculated over such a space.
Furthermore, our attack can even be extended to the Hamming distance with the following idea:
for $d$-dimensional binary vectors $x,y \in \{0, 1\}^{d}$, let us encode them to vectors $\hat{x}, \hat{y} \in  \{-\frac{1}{\sqrt{d}},\frac{1}{\sqrt{d}}\}^{d}$ by mapping $0\mapsto -\frac{1}{\sqrt{d}}$ and $1 \mapsto \frac{1}{\sqrt{d}}$ for each component. 
Then $\|\widehat{x}\|_{2} = \|\widehat{y}\|_{2}=1$ and $\mathsf{HD}(x,y) = d \cdot \frac{1 - \langle \hat{x}, \hat{y} \rangle}{2}$ holds, where $\mathsf{HD}$ denotes the Hamming distance.
As a result, we can recover the cosine similarity between encoded features from the Hamming distance score queries.
Thus, with appropriate score-based attacks, our attack can be applied to biometric modalities using Euclidean or Hamming distances as a distance metric, including iris recognition based on IrisCode~\cite{daugman2009iris} or deep learning-based speaker recognition~\cite{jung2022pushing}.

\subsubsection*{\bf Connection with Intrinsic Entropy of Biometrics}
Conceptually, FMR is linked to the intrinsic entropy of the biometric traits, which the template extractor aims to capture.
From an information-theoretic perspective, fuzzy extractors~\cite{dodis2004fuzzy} formalize this connection by relating intrinsic entropy to statistical security guarantees, quantifying how many bits of nearly uniform randomness can be extracted from some random source.
In particular, fuzzy min entropy~\cite{fuller2016fuzzy} refines this idea by measuring this entropy as the probability mass of a metric ball, which is closely related to the FMR.
Capturing such non-uniformity yields a non-negligible FMR, as opposed to the fact that spherical caps occupy negligible area over the high-dimensional hypersphere.
Our \net-based attack can be viewed as drawing a multi-query analogy of fuzzy min entropy, in the sense that \net corresponds to the metric balls whose union concentrates a large probability mass under the template distribution.
Notably, our experimental results suggest that such a construction captures non-trivial structure beyond that implied by fuzzy min entropy alone, as evidenced by attack success rates exceeding FMR.
This observation points to a potential connection between \net and existing impossibility results in fuzzy extractors based on fuzzy min entropy~\cite{fuller2016fuzzy, fuller2024impossibility}.
We believe that this viewpoint would deepen our understanding of the intrinsic entropy of biometrics, leaving its formal investigation as an interesting direction for future work.

% \red{SH: Need to be revised}

\section{Related Work}

% \input{figs/Tab1}

% \red{To be revised.}

We survey existing attacks against FRSs in terms of adversarial goals and capabilities, highlighting why their adversarial settings do not align with the target scenarios considered in this work.
% We provide a summary of these attacks in terms of adversarial goal and capability in Table~\ref{tab:comp_table}.

\subsubsection*{\bf Adversarial Attacks}
Adversarial attacks aim to maliciously modify the input to induce the machine learning model to make wrong decisions.
Such a modified input is referred to as an adversarial example.
In particular, adversarial attacks against FRSs have been actively studied, targeting commercial APIs.
Some studies focused on crafting adversarial examples by leveraging the score~\cite{kim2025non} or decision~\cite{dong2019efficient, zheng2021black, chen2020boosting} results of two submitted biometric samples.
Some attacks require direct score queries to the enrolled identity~\cite{chen2021real, du2020sirenattack, sharif2016accessorize}, utilizing zeroth-order optimization algorithms to find adversarial examples.
Another line of work~\cite{li2023sibling, jia2022adv, salar2025enhancing} utilizes surrogate models, relying on the transferability of adversarial examples.

However, these attacks require that the adversary can either (i) choose which biometric sample to be enrolled~\cite{kim2025non, dong2019efficient, li2023sibling, jia2022adv, salar2025enhancing, zheng2021black} or (ii) obtain a similarity score between enrolled and queried samples~\cite{chen2021real, du2020sirenattack}.
Although their adversarial models are well-justified and considered feasible in various scenarios, e.g., generating fraudulent identities for romance scams~\cite{kim2025non}, from the perspective of impersonation attacks, their strong adversarial capabilities are not applicable in practical impersonation scenarios.

\subsubsection*{\bf Reconstruction Attacks}
Reconstruction attacks focus on extracting the biometric characteristics of the enrolled identity through interactions with the target FRS.
Some methods model the attack as an optimization problem where the adversary can directly obtain the scores between the enrolled identity and queried samples in the target FRS.
To address this, several attack methods aim to design corresponding solvers, such as the hill-climbing algorithm~\cite{razzhigaev2021darker, vendrow2021realistic, jung2024face}, or utilize the geometric structure of the template space~\cite{kim2024scores}.
In contrast, some attacks tried to directly train a neural network that reconstructs the face corresponding to the input template, assuming that the adversary can query faces to obtain the corresponding template extracted from the target FRS~\cite{mai2018reconstruction, duong2020vec2face, lu2021voxstructor, shahreza2024vulnerability}.
Afterwards, if the adversary compromises the target FRS's database to obtain the target template, such as via a data breach, the adversary can reconstruct the target identity's face.

In our attack, we can leverage score-based attacks to recover faces corresponding to \net. 
Remark that we utilized the Kim et al.~\cite{kim2024scores} attack as a sub-routine to find an initial point of the subsequent score-based attack that utilizes a latent space search algorithm~\cite{vendrow2021realistic, jung2024face}.
However, similar to adversarial attacks, these attacks per se are not applicable in our threat model as they require powerful adversarial capabilities, e.g., direct score or template queries, or require obtaining the target template stored in the target FRS.

\subsubsection*{\bf Morphing Attack and MasterFaces}
Morphing attacks~\cite{damer2018morgan, colbois2023approximating} and Wolf attacks~\cite{une2007wolf, otsuka2013wolf} aim to construct biometric samples that are accepted as two or more distinct identities.
MasterFaces~\cite{nguyen2020generating, nguyen2022master, shmelkin2021generating, friedlander2022generating, terhorst2022limited} instantiate the Wolf attack paradigm in the face recognition domain.
Such faces effectively act as a backdoor in the target FRS; once enrolled, all identities contributed to crafting such faces can be successfully authenticated.
However, from the perspective of impersonation attacks, these attacks are not directly applicable, as they do not target a victim identity enrolled in the target system.

Several studies have explored the use of MasterFaces for impersonation, from the intuition that they are more plausible to impersonate the target identity than randomly sampled faces.
However, as we summarized in Section~\ref{sec:OurAttack}, these attacks are ineffective against modern FRSs due to their heavy reliance on the template extraction pipeline and suboptimality of their algorithm, and thus they are believed not to pose a serious threat against modern FRSs.

The proposed attack shares the same intuition as MasterFaces, but overcomes their limitations by introducing a new \net-based formulation and Gram decomposition technique that enables the adversary to effectively exploit API queries legitimately.
From this perspective, our method repositions MasterFace impersonation attacks as a potential security threat when the adversary can purchase and utilize pay-as-you-go commercial API services.

\subsubsection*{\bf Other Related Attacks}
Alongside the above classes of attacks, other attacks also reveal various vulnerabilities of the target FRS. 
Some studies attempted to directly launch an impersonation attack through a hill-climbing algorithm via score queries~\cite{maiorana2014hill, galbally2010vulnerability, jeong2022analysis, an2023imu}.
Likewise, score-based attacks in adversarial and reconstruction attacks, their adversarial assumptions do not align with the practical constraints.
Mimicry attacks or replay attacks have also been considered~\cite{khan2020mimicry, yoon2020new, kinnunen2019can}, which submit the recorded biometric signals of the enrolled users, exploiting physical or sensor-level vulnerabilities.

\subsubsection*{\bf Takeaway}
While several attacks against FRSs have been proposed from various approaches and attack surfaces, under the threat model where highly restricted forms of interactions with the target FRS are permitted, most of these attacks become inapplicable or fail to achieve a non-trivial impersonation rate.
This is the reason why the zero-effort impostor remains a standalone baseline; in contrast, our attack demonstrates that the adversary can surpass this barrier by purchasing a publicly disclosed commercial API.

\section{Conclusion}

This work demonstrates that an adversary's legitimate API queries enable a non-trivial impersonation attack beyond the zero-effort impostor baseline, even under a realistic-yet-constrained threat model in which almost all existing attack techniques are inapplicable.
This result is enabled by our two core ingredients: \net-based formulation and Gram decomposition technique, which effectively overcome the limitations of MasterFaces-based attacks, and thus revive them as a potential security threat against real-world FRS deployments.
Specifically, we expect that our Gram decomposition technique will serve as a useful tool for discovering other types of vulnerabilities where the adversary aims to extract non-trivial information about the unknown FRS itself through interactions, such as membership inference or model extraction attacks.
We believe that our findings open new directions for assessing the security risks associated with the public disclosure of commercial APIs.

\section*{Ethical Statement}
We attest that we thoroughly read the ethics discussions in the conference call for papers, the detailed submission instructions, and the ethics guidelines.
We attest that the research team considered the ethics of this research, that the authors believe the research was done ethically, and that the team's next step plans are ethical.
More precisely, we utilized public benchmark datasets for experiments; no private or sensitive data was collected by the authors.
We provide details about the public datasets we used, including (i) how each dataset was collected, (ii) how each dataset was processed for evaluation, and (iii) how to obtain these datasets.
Since only public data were used, IRB approval was not required.
We attest that we discovered and investigated the proposed attack for research only, and do not conduct an actual attack against real-world systems.
To mitigate potential misuse, we will release only the minimal code necessary for reproducing results.

\section*{Open Sciences}
We provide the source code through the Zenodo archive: \url{https://zenodo.org/records/20765343}.
Given the potential societal harm, we only provide the minimal attack pipeline enough to validate the correctness, and the released implementation is not sufficient for directly attacking real-world systems.
We also note that our artifact contains how to obtain real-world datasets and open-source FRSs we used during the experiments, each of which is publicly available.

\begin{acks}
We thank the anonymous reviewers for their helpful comments and feedback.
This work was supported in part by the Institute of Information and Communications Technology Planning and Evaluation (IITP), grant funded by the
Korea Government (MSIT) (RS-2021-II210727), and the Culture, Sports and Tourism Research and Development Program through Korea Creative Content Agency Grant funded by the Ministry of Culture, Sports and Tourism, under Grant RS-2024-00332210.
\end{acks}

\bibliographystyle{ACM-Reference-Format}
\bibliography{bib_short}

\appendix

\section{Detailed Settings and Additional Results}\label{sec:ExpDetail}

\subsubsection*{\bf Dataset Selection Rationale.}

As summarized in Table~\ref{tab:datasets}, we select datasets by jointly considering the number of facial images and the number of identities, depending on their roles in our evaluation.
For constructing the \net, we choose MS1MV3~\cite{DBLP:conf/iccvw/DengGZDLS19}, which is one of the largest publicly available face datasets and contains both a massive number of images and a sufficiently large number of identities. This choice enables the \net to effectively cover a broad region of the template space.
As target datasets, we select CASIA-WebFace~\cite{yi2014learning} and LFW~\cite{huang2008labeled}, which are smaller than MS1MV3 but still include a sufficiently large number of identities to support reliable impersonation evaluation.
In contrast, AgeDB-30~\cite{moschoglou2017agedb} and CFP-FP~\cite{sengupta2016frontal} are commonly used as verification benchmarks in face recognition. However, due to their relatively small numbers of identities, they are not well-suited for our fitting and coverage analysis, as the results could be overly optimistic, pessimistic, or biased. Therefore, we do not include them as primary target datasets in the main experiments. Instead, they are used for additional performance evaluation of the proposed and baseline models, together with LFW, to assess generalization across standard verification benchmarks.

\begin{table}[b]
\centering
\vspace{-5pt}
\caption{Summary of face datasets used for constructing the \net, target enrolled identities, and additional evaluations.}
\vspace{-10pt}
\resizebox{\linewidth}{!}{
\begin{tabular}{c|c|cc|cc}
\hline
 & \textbf{\net} & \multicolumn{2}{c|}{\textbf{Target}} & \multicolumn{2}{c}{\textbf{Additional}} \\
\cline{2-6}
\textbf{Dataset} 
& MS1MV3 
& CASIA-WebFace 
& LFW 
& AgeDB-30 
& CFP-FP \\
\hline
\# Images 
& 5,179,510 
& 490,623
& 13,233 
& 16,488 
& 7,000 \\ \hline
\# Ids 
& 93,431 
& 10,572 
& 5,749 
& 568 
& 500 \\
\hline
\end{tabular}
}
\label{tab:datasets}
% \vspace{-10pt}
\end{table}

\subsubsection*{\bf More Details on Open-sourced Models}
We provide detailed information about open-source FRSs and AWS CompareFace API, including the decision thresholds for each FMR level and the verification benchmark results.
We measure the threshold by following the procedure in Section~\ref{sec:expsetting} using the MS1MV3 dataset.
Specifically, for the AWS CompareFace, we randomly select 8,192 identities and determine the threshold by querying all pairs of the corresponding faces.
Note that AWS CompareFace provides confidence scores ranging from 0 to 100; a higher score means the given pair of images is similar. 
For this reason, we report the decision threshold in terms of those scores.
The thresholds are provided in Tab.~\ref{tab:thx}.

\begin{table}[t]
    \centering
\caption{Thresholds of FRSs for FMR level}
\vspace{-10pt}
    \begin{tabular}{c|c|c|c}
        \hline
        FMR & $10^{-3}$ & $10^{-4}$ & $10^{-5}$ \\ \hline \hline
        $F_1$ & 0.1881 & 0.2413 & 0.2995 \\ \hline
        $F_2$ & 0.1818 & 0.2313 & 0.2857 \\ \hline
        $F_3$ & 0.1852 & 0.2406 & 0.3008 \\ \hline
        $F_4$ & 0.2258 & 0.2856 & 0.3434 \\ \hline
        $F_\mathsf{A}$ & 22.8715 & 45.0380 & 68.7471 \\ \hline
\end{tabular}
\vspace{-5pt}
\label{tab:thx}
\end{table}

\begin{table}[t]
\centering
\caption{Verification performance (TAR and FAR) at operational FMR levels on three benchmarks. 
For each FRS, the decision threshold is set to match the target FMR, and the resulting TAR(\%)/FAR(\%) are reported on each dataset.}
\vspace{-10pt}
\label{tab:tar_far}
\resizebox{0.9\linewidth}{!}{
\begin{tabular}{c|cc|cc|cc}
\hline
\multirow{2}{*}{\textbf{FRS}} 
& \multicolumn{2}{c|}{\textbf{LFW}} 
& \multicolumn{2}{c|}{\textbf{AgeDB-30}} 
& \multicolumn{2}{c}{\textbf{CFP-FP}} \\
\cline{2-7}
& \textbf{TAR} & \textbf{FAR} 
& \textbf{TAR} & \textbf{FAR} 
& \textbf{TAR} & \textbf{FAR} \\
\hline\hline
\multicolumn{7}{c}{\textbf{FMR = $10^{-3}$}} \\ \hline
$F_1$ & 99.7 & 0.167 & 97.267 & 0.767 & 98.771 & 0.171 \\
$F_2$ & 99.7 & 0.2 & 96.933 & 0.367 & 97.429 & 0.2 \\
$F_3$ & 99.733 & 0.133 & 96.933 & 0.8 & 98.714 & 0.171 \\
$F_4$ & 99.7 & 0.133 & 96.833 & 0.833 & 98.743 & 0.143 \\
$F_{\mathsf{A}}$ & 99.7 & 0.233 & 97.1 & 0.367 & 99.171 & 0.086 \\
\hline
\multicolumn{7}{c}{\textbf{FMR = $10^{-4}$}} \\ \hline
$F_1$ & 99.7 & 0.033 & 96.0 & 0.033 & 98.457 & 0 \\
$F_2$ & 99.7 & 0.033 & 95.867 & 0.1 & 95.343 & 0.057 \\
$F_3$ & 99.667 & 0.033 & 94.767 & 0.033 & 97.857 & 0 \\
$F_4$ & 99.667 & 0 & 95.533 & 0.167 & 98.2 & 0 \\
$F_{\mathsf{A}}$ & 99.7 & 0.033 & 95.6 & 0 & 98.4 & 0.029 \\
\hline
\multicolumn{7}{c}{\textbf{FMR = $10^{-5}$}} \\ \hline
$F_1$ & 99.633 & 0 & 93.967 & 0 & 97.514 & 0 \\
$F_2$ & 99.6 & 0 & 93.667 & 0 & 91.657 & 0 \\
$F_3$ & 99.467 & 0 & 90.833 & 0 & 95.8 & 0 \\
$F_4$ & 99.667 & 0 & 92.933 & 0 & 97.0 & 0 \\
$F_{\mathsf{A}}$ & 99.667 & 0.033 & 93.667 & 0 & 96.857 & 0 \\
\hline
\end{tabular}
}
\vspace{-10pt}
\end{table}

In addition to this, we also provide the verification benchmark results, including LFW, AgeDB-30, and CFP-FP, for the threshold determined above.
More precisely, we regard the true acceptance rate (TAR) and false acceptance rate (FAR)\footnote{Note that these quantities originally consider system-level failures, e.g., rejections from anti-spoofing methods, according to the standards~\cite{grother2013biometric}. Nevertheless, we will use these terms in the same meaning of true match rate and false match rate, following the convention based on the LFW's standard evaluation protocol~\cite{huang2008labeled}.} as the evaluation metrics and measure them at the threshold per FMR level.
The results are provided in Tab.~\ref{tab:tar_far}.
We can observe that, because of the distributional discrepancy between MS1MV3 and each benchmark dataset, the FAR at each benchmark is slightly different from the FMR value corresponding to each threshold.
Nevertheless, for tighter thresholds corresponding to FMR$=10^{-4}$ or $10^{-5}$, we obtain almost the same result as expected.

\subsubsection*{\bf Additional Experimental Results}
We provide additional experimental results on measuring IRs for open-source FRSs $F_{3}$ and $F_{4}$, which were omitted in the main text due to space constraints.
The results are provided in Tab.~\ref{tab:open_f3} and Tab.~\ref{tab:open_f4}, respectively.
Consistent with the results in $F_{2}$ and $F_{\mathsf{A}}$, the obtained IR values from our attack consistently exceed those from both the zero-effort imposter baseline and those estimated from FMR.
Particularly, as we briefly mentioned in the main text, they exhibit a similar behavior as those measured from $F_{2}$ in Tab.~\ref{tab:open}, i.e., our attack becomes (relatively) more effective than the baseline for smaller authentication trial budgets and stricter FMR levels.

\section{Details on Score-based Attacks}\label{sec:SBDetail}

We provide details of the score-based attacks we used in our experiments, which were omitted in the main text.

\subsubsection*{\bf Kim et al.'s Attack~\cite{kim2024scores}}
The core idea of their attack stems from the observation that well-trained FRSs act like almost isometries, i.e., the distance relationship between templates from facial images is largely preserved across the template extractors.
This motivates the adversary to view the score between the queried faces and the enrolled one in the system as a proxy for the location of the corresponding template over the target model's template space.
Specifically, thanks to the observation mentioned above, the adversary can estimate the location of the target face's template over the template space of its local surrogate model.

More formally, for the target face $\mathfrak{b}$ and the queried ones $\mathfrak{v}_{1}, \dots, \mathfrak{v}_{Q}$, let us denote $F_{T}$ and $F_{L}$ as the template extractors of the target FRS and the adversary's surrogate model, respectively.
Then their observation tells us that 
\begin{align}
    \langle F_{T}(\mathfrak{v_{i}}), F_{T}(\mathfrak{b}) \rangle \approx \langle F_{L}(\mathfrak{v}_{i}), F_{L}(\mathfrak{b}) \rangle \text{ for } i=1,\dots, Q \nonumber.
\end{align}
That is, if we denote $A$ as a matrix whose $i^{\text{th}}$ column vector is $F_{L}(\mathfrak{v}_{i})$, we obtain that
\begin{align}
    A^{T}F_{L}(\mathfrak{b}) \approx ( \langle F_{T}(\mathfrak{v_{i}}), F_{T}(\mathfrak{b}) \rangle )_{i=1}^{Q}. \nonumber
\end{align}

\begin{table}[t]
\centering
\caption{IR (\%) results on $F_3$ at various FMR levels.  
For \textbf{Ours}, the second row in each cell, highlighted as \textbf{\color{blue}bold}, indicates the improvement factor over the baseline.}
\vspace{-10pt}
\resizebox{1.0\linewidth}{!}{
\begin{tabular}{c|c|ccc|ccc}
\hline
\multirow{3}{*}{$Q$} & \multirow{3}{*}{Method} & \multicolumn{6}{c}{Target Dataset} \\ \cline{3-8} 
 &  & \multicolumn{3}{c|}{LFW} & \multicolumn{3}{c}{CASIA-WebFace} \\ \cline{3-8} 
 & & $10^{-3}$ & $10^{-4}$& \multicolumn{1}{c|}{$10^{-5}$} & $10^{-3}$ & $10^{-4}$& $10^{-5}$ \\ \hline \hline
\multirow{4}{*}{$5$}
 & $Q \cdot \mathsf{FMR}$ 
 & 0.5 & 0.05 & 0.005 & 0.5 & 0.05 & 0.005 \\ \cline{2-8}

 & Baseline 
 & 0.495 & 0.048 & 0.007 & 0.495 & 0.048 & 0.007 \\ \cline{2-8}

 & Ours
 & \begin{tabular}{c} 2.901 \\ \textbf{\color{blue}$\times$5.86} \end{tabular}
 & \begin{tabular}{c} 0.388 \\ \textbf{\color{blue}$\times$8.08} \end{tabular}
 & \begin{tabular}{c} 0.035 \\ \textbf{\color{blue}$\times$5.00} \end{tabular}
 & \begin{tabular}{c} 2.803 \\ \textbf{\color{blue}$\times$5.66} \end{tabular}
 & \begin{tabular}{c} 0.408 \\ \textbf{\color{blue}$\times$5.66} \end{tabular}
 & \begin{tabular}{c} 0.019 \\ \textbf{\color{blue}$\times$2.71} \end{tabular}
 \\ \hline

\multirow{4}{*}{$10$}
 & $Q \cdot \mathsf{FMR}$ 
 & 1 & 0.1 & 0.01 & 1 & 0.1 & 0.01 \\ \cline{2-8}

 & Baseline 
 & 0.986 & 0.095 & 0.013 & 0.988 & 0.096 & 0.014 \\ \cline{2-8}

 & Ours
 & \begin{tabular}{c} 4.730 \\ \textbf{\color{blue}$\times$4.80} \end{tabular}
 & \begin{tabular}{c} 0.646 \\ \textbf{\color{blue}$\times$6.80} \end{tabular}
 & \begin{tabular}{c} 0.043 \\ \textbf{\color{blue}$\times$3.31} \end{tabular}
 & \begin{tabular}{c} 5.091 \\ \textbf{\color{blue}$\times$5.15} \end{tabular}
 & \begin{tabular}{c} 0.794 \\ \textbf{\color{blue}$\times$8.27} \end{tabular}
 & \begin{tabular}{c} 0.070 \\ \textbf{\color{blue}$\times$5.00} \end{tabular}
 \\ \hline

\multirow{4}{*}{$30$}
 & $Q \cdot \mathsf{FMR}$ 
 & 3 & 0.3 & 0.03 & 3 & 0.3 & 0.03 \\ \cline{2-8}

 & Baseline 
 & 2.928 & 0.286 & 0.039 & 2.933 & 0.288 & 0.040 \\ \cline{2-8}

 & Ours
 & \begin{tabular}{c} 10.74 \\ \textbf{\color{blue}$\times$3.67} \end{tabular}
 & \begin{tabular}{c} 1.510 \\ \textbf{\color{blue}$\times$5.28} \end{tabular}
 & \begin{tabular}{c} 0.157 \\ \textbf{\color{blue}$\times$4.03} \end{tabular}
 & \begin{tabular}{c} 12.29 \\ \textbf{\color{blue}$\times$4.19} \end{tabular}
 & \begin{tabular}{c} 2.225 \\ \textbf{\color{blue}$\times$7.73} \end{tabular}
 & \begin{tabular}{c} 0.309 \\ \textbf{\color{blue}$\times$7.73} \end{tabular}
 \\ \hline
\end{tabular}
}
\vspace{-10pt}
\label{tab:open_f3}
\end{table}

\begin{table}[t]
\centering
\caption{IR (\%) results on $F_4$ at various FMR levels.  
For \textbf{Ours}, the second row in each cell, highlighted as \textbf{\color{blue}bold}, indicates the improvement factor over the baseline.}
\vspace{-10pt}
\resizebox{1.0\linewidth}{!}{
\begin{tabular}{c|c|ccc|ccc}
\hline
\multirow{3}{*}{$Q$} & \multirow{3}{*}{Method} & \multicolumn{6}{c}{Target Dataset} \\ \cline{3-8} 
 &  & \multicolumn{3}{c|}{LFW} & \multicolumn{3}{c}{CASIA-WebFace} \\ \cline{3-8} 
 & & $10^{-3}$ & $10^{-4}$& \multicolumn{1}{c|}{$10^{-5}$} & $10^{-3}$ & $10^{-4}$& $10^{-5}$ \\ \hline \hline
\multirow{4}{*}{$5$}
 & $Q \cdot \mathsf{FMR}$ 
 & 0.5 & 0.05 & 0.005 & 0.5 & 0.05 & 0.005 \\ \cline{2-8}

 & Baseline 
 & 0.468 & 0.046 & 0.007 & 0.526 & 0.058 & 0.010 \\ \cline{2-8}

 & Ours
 & \begin{tabular}{c} 1.428 \\ \textbf{\color{blue}$\times$3.05} \end{tabular}
 & \begin{tabular}{c} 0.175 \\ \textbf{\color{blue}$\times$3.80} \end{tabular}
 & \begin{tabular}{c} 0.014 \\ \textbf{\color{blue}$\times$2.00} \end{tabular}
 & \begin{tabular}{c} 2.086 \\ \textbf{\color{blue}$\times$3.97} \end{tabular}
 & \begin{tabular}{c} 0.373 \\ \textbf{\color{blue}$\times$6.43} \end{tabular}
 & \begin{tabular}{c} 0.051 \\ \textbf{\color{blue}$\times$5.10} \end{tabular}
 \\ \hline

\multirow{4}{*}{$10$}
 & $Q \cdot \mathsf{FMR}$ 
 & 1 & 0.1 & 0.01 & 1 & 0.1 & 0.01 \\ \cline{2-8}

 & Baseline 
 & 0.934 & 0.091 & 0.013 & 1.049 & 0.116 & 0.022 \\ \cline{2-8}

 & Ours
 & \begin{tabular}{c} 2.843 \\ \textbf{\color{blue}$\times$3.04} \end{tabular}
 & \begin{tabular}{c} 0.324 \\ \textbf{\color{blue}$\times$3.56} \end{tabular}
 & \begin{tabular}{c} 0.024 \\ \textbf{\color{blue}$\times$1.85} \end{tabular}
 & \begin{tabular}{c} 3.615 \\ \textbf{\color{blue}$\times$3.45} \end{tabular}
 & \begin{tabular}{c} 0.686 \\ \textbf{\color{blue}$\times$5.91} \end{tabular}
 & \begin{tabular}{c} 0.100 \\ \textbf{\color{blue}$\times$4.55} \end{tabular}
 \\ \hline

\multirow{4}{*}{$30$}
 & $Q \cdot \mathsf{FMR}$ 
 & 3 & 0.3 & 0.03 & 3 & 0.3 & 0.03 \\ \cline{2-8}

 & Baseline 
 & 2.773 & 0.273 & 0.040 & 3.111 & 0.348 & 0.064 \\ \cline{2-8}

 & Ours
 & \begin{tabular}{c} 7.242 \\ \textbf{\color{blue}$\times$2.61} \end{tabular}
 & \begin{tabular}{c} 0.897 \\ \textbf{\color{blue}$\times$3.29} \end{tabular}
 & \begin{tabular}{c} 0.065 \\ \textbf{\color{blue}$\times$1.63} \end{tabular}
 & \begin{tabular}{c} 9.130 \\ \textbf{\color{blue}$\times$2.94} \end{tabular}
 & \begin{tabular}{c} 1.619 \\ \textbf{\color{blue}$\times$4.65} \end{tabular}
 & \begin{tabular}{c} 0.202 \\ \textbf{\color{blue}$\times$3.16} \end{tabular}
 \\ \hline
\end{tabular}
}
% \vspace{-10pt}
\label{tab:open_f4}
\end{table}

Although the above linear system is underdetermined, utilizing the least-squares solver, i.e., pseudoinverse, gives a sufficiently close solution to $F_{L}(\mathfrak{b})$.
In particular, to reduce the effect of the error in the score, they make the templates $\{F_{L}(\mathfrak{v}_{i})\}_{i=1}^{Q}$ of the queried faces as close to orthogonal as possible.
They call such a set of faces an orthogonal face set (OFS), along with providing an algorithm to find an OFS.
Finally, the candidate face template from least squares is fed to the inverse model $F^{\dagger}_{L}$ of $F_{L}$ held by the adversary, which reconstructs the corresponding faces to the given template.
We provide the detailed algorithm of this attack in Algorithm~\ref{alg:KimAttack}.

We remark that their attack is non-adaptive, in the sense that the queried faces (OFS) are not updated during the query.
Particularly, the choice of the OFS depends on the adversary's surrogate model only.
For this reason, when employing Kim et al.'s attack in our attack framework, we can directly select $\mathtt{DB}_{p}$ as such an OFS, and utilize their rotated templates to recover the faces corresponding to the rotated \net.

\begin{algorithm}[t]
\caption{Kim et al.'s Reconstruction Attack}\label{alg:KimAttack}
\begin{algorithmic}[1]
\REQUIRE Score oracle $\mathcal{O}(\cdot)$, OFS $\{\mathfrak{v}_{1}, \dots, v_{Q}\}$, local surrogate template extractor $F_{L}$ and its inverse model $F_{L}^{\dagger}$.
\ENSURE A facial image $\mathfrak{b} \in \mathcal{B}$
\STATE Compute a matrix $A$ whose $i^{\text{th}}$ column vector is $F_{L}(\mathfrak{v}_{i})$.
\STATE Obtain a score vector $\vec{s}=(\mathcal{O}(\mathfrak{v}_{i}))_{i=1}^{Q}$ from oracle queries.
\STATE Compute $z\gets A^{\dagger}\vec{s}$, where $A^{\dagger}$ is the pseudoinverse of $A$.
\RETURN $\mathfrak{b} \gets F^{\dagger}_{L}(\frac{z}{\|z\|_{2}})$.
\end{algorithmic}
\end{algorithm}

\subsubsection*{\bf NES-based Post-Processing}
After recovering faces from Kim et al.'s attack, we apply a post-processing optimization algorithm based on natural evolution strategy (NES)~\cite{wierstra2014natural} to recover more accurate MasterFaces.
Since the initial point found by Kim et al.'s attack aligns with the latent space of their inverse model, we use the latent search variant of the NES.
In particular, we utilize a hypersphere variant of the NES during gradient estimation as the latent space of the inverse model is also defined as $\mathbb{S}^{d-1}$.

For the given initial point $z_{0} \in \mathbb{S}^{d-1}$, the algorithm is processed as follows: first, the adversary samples random vectors, say $u_{1}, \dots, u_{n_{it}}$, from the uniform distribution over the hypersphere $\mathcal{U}(\mathbb{S}^{d-1})$, where $n_{it}$ denotes the number of random samples for the gradient estimation.
The adversary then project each $u_{i}$ to its tangent space by computing $u_{i} \gets u_{i} - \langle u_{i}, z_{0} \rangle z_{0}$ and $\tilde{u_{i}} \gets \frac{u_{i}}{\|u_{i}\|_{2}}$
For the step size $\theta$, the gradient estimation is processed by perturbing $z_{0}$ from the direction of $u_{i}$ by $\theta$ following the geodesic line over $\mathbb{S}^{d-1}$.
More precisely, the adversary computes
\begin{align}
    z^{+}_{i} = \cos\theta \cdot z_{0} + \sin\theta \cdot \tilde{u}_{i}, \quad 
    z^{-}_{i} = \cos\theta \cdot z_{0} - \sin\theta \cdot \tilde{u}_{i} \nonumber.
\end{align}
Here, computing both $z^{+}_{i}$ and $z^{-}_{i}$ corresponds to the standard antithetic sampling, which reduces the variance of the estimated gradient.
Finally, the adversary obtains the estimated gradient through $2 \cdot n_{it}$ calls of the score oracle $\mathcal{O}(\cdot )$, namely,
\begin{align}
    s_{i}^{+} \gets \mathcal{O}(F^{\dagger}(z_{i}^{+})), \quad  s_{i}^{-} \gets \mathcal{O}(F^{\dagger}(z_{i}^{-})), \quad g\gets \frac{1}{n_{it}}\sum_{i=1}^{n_{it}}\frac{s_{i}^{+} - s_{i}^{-}}{2 \sin \theta}\tilde{u}_{i} \nonumber.
\end{align}
The final $g$ is used for applying the gradient descent on $z_{0}$, namely, for the learning rate $\eta$, $z_{0}$ is updated to $z_{0} + \eta \cdot g$, i.e., the direction that maximizes the confidence score.
Note that the adversary can apply a momentum-based approach to the gradient to stabilize the optimization trajectory.
We provide the algorithm for the single-step update explained above in Algorithm~\ref{alg:HNES}.

\begin{algorithm}[t]
\caption{Hypersphere NES}\label{alg:HNES}
\begin{algorithmic}[1]
\REQUIRE Score oracle $\mathcal{O}(\cdot)$, initial point $z_{0} \in \mathbb{S}^{d-1}$, \# of samples $n_{it}$ for gradient estimation, learning rate $\eta$, search step $\theta$, the inverse model $F^{\dagger}$
\ENSURE Updated point $z \in \mathbb{S}^{d-1}$.
\FOR{$i$ from 1 to $n_{it}$}
\STATE Sample $u_{i} \gets \mathcal{U}(\mathbb{S}^{d-1})$
\STATE Compute $u_{i} \gets u_{i}- \langle{u_{i},z_{0}} \rangle z_{0}$ and $\tilde{u_{i}} \gets \frac{u_{i}}{\|u_{i}\|_{2}}$
\STATE Compute $z_{i}^{+} \gets  z_{0}\cos \theta  + \tilde{u_{i}}\sin\theta$, $z_{i}^{-} \gets z_{0}\cos \theta - \tilde{u_{i}}\sin\theta$.
\STATE Obtain $s_{i}^{+} \gets \mathcal{O}(F^{\dagger}(z_{i}^{+}))$ and $s_{i}^{-} \gets \mathcal{O}(F^{\dagger}(z_{i}^{-}))$ from queries.
\ENDFOR
\STATE Compute $g \gets \frac{1}{n_{it}}\sum_{i=1}^{n_{it}} \frac{s^{+}_{i} - s^{-}_{i}}{2\sin\theta} u_{i}$.
\STATE Compute $z \gets z_{0} + \eta \cdot g$ and \textbf{return} $z \gets \frac{z}{\|z\|_{2}}$
\end{algorithmic}
\end{algorithm}

\subsubsection*{\bf Full Attack Parameters}
We provide the full parameters for the score-based attacks used in our experiments.
\begin{itemize}[leftmargin=*]
    \item \textbf{Kim et al.'s Attack}
    \begin{itemize}
        \item By following the subsequent work~\cite{kim2025non}, we construct the OFS as the faces corresponding to the top 100 principal components of the templates in the WebFace42M dataset~\cite{zhu2021webface260m}. We use the first image for each identity to run the PCA.
        \item We use the same inverse model, modified NbNet~\cite{kim2024scores}.
    \end{itemize}

    \item \textbf{Hyperspherical NES}
    \begin{itemize}
        \item We use $n_{it} = 1$ as it provides stable convergence when combined with momentum of ratio $0.1$.
        
        \item We dynamically update $\theta$ as follows: for $\theta_{0}=0.05$, we set $\theta = \frac{\theta_{0}}{1 + \|g\|_{2}}$ and clamp it between $0.01$ and $0.15$.

        \item The total number of updates depends on the query budget; we set $k \cdot 200$ as the maximum budget.
    \end{itemize}    
\end{itemize}

\section{Omitted Proofs}\label{sec:Proofs}

We provide the omitted proofs of our arguments, especially for the dimensionality reduction on the Gram decomposition in Section~\ref{sec:detail}.
We first present the following lemma, a toolbox to analyze the transformed templates obtained from the Gram decomposition.

\begin{lemma}\label{lemma:First}
Let $A \in \mathbb{R}^{d \times k}$ be a matrix whose singular value decomposition is $U\Sigma V^{T}$ for $U \in \mathbb{R}^{d \times k}$ and $\Sigma, V \in \mathbb{R}^{k \times k}$, and $L \in \mathbb{R}^{k \times k}$ be a Gram decomposition of $A^{T}A$, i.e., $L^{T}L=A^{T}A$.
Then the following properties hold:
\begin{itemize}[leftmargin=*]
    \item There exists an unitary matrix $R \in \mathbb{R}^{k\times k}$ such that $L = RU^{T}A$.

    \item For $z \in \mathbb{R}^{d}$, $RU^{T}z = (z^{T}A \cdot L^{-1})^{T}$. In particular, $z^{T}A = (RU^{T}z)^{T}L$.

    \item For $x, y\in \mathbb{R}^{d}$, $\langle RU^{T}x, \frac{RU^{T}y}{\|RU^{T}y\|_{2}} \rangle = \langle x, \frac{AA^{\dagger}y}{\|AA^{\dagger}y\|_{2}} \rangle$,
\end{itemize}
where $A^{\dagger}$ is the pseudoinverse of $A$. 
\end{lemma}

\begin{proof}
Observe that $A^{T}A = V\Sigma^{T} U^{T}U \Sigma V^{T}$.
Since $U^{T}U = I_{k}$ and $\Sigma^{T} = \Sigma$, we can deduce that $L^{T}L=A^{T}A= (\Sigma V)^{T}(\Sigma V)$ for $\Sigma V^{T} \in \mathbb{R}^{k \times k}$
That is, the unitary ambiguity of Gram decomposition ensures the existence of the unitary matrix $R \in \mathbb{R}^{k \times k}$ such that $L = R\Sigma V^{T}$. Finally, since $U^{T}A = (U^{T}U)\Sigma V^{T} = \Sigma V^{T}$, we finally obtain $L = RU^{T}A$, thereby showing the first property.

For the second property, we first note that $UU^{T}$ corresponds to the projection matrix onto the $A$'s column space.
That is, $UU^{T}A = A$ holds.
From this observation, we can deduce that
\begin{align}
    z^{T}A= z^{T}UU^{T}A = z^{T}UR^{T}RU^{T}A = (RU^{T}z)^{T}(RU^{T}A) \nonumber.
\end{align}
Hence, provided that $L=RU^{T}A$ is invertible, we finally obtain
\begin{align}
    RU^{T}z = \Big( (z^{T}A) \cdot (RU^{T}A)^{-1} \Big)^{T} = (z^{T}A  \cdot L^{-1})^{T} \nonumber.
\end{align}
This proves the second property.

To show the third property, we first observe that the following equality holds
\begin{align}
    \langle RU^{T}x, RU^{T}y \rangle =  \langle U^{T}x, U^{T}y \rangle = \langle x, UU^{T}y \rangle \nonumber,
\end{align}
because $R$ is an isometry over $\mathbb{R}^{k}$.
Here, $AA^{\dagger} = U\Sigma V^{T}V \Sigma^{-1} U^{T} = AA^{\dagger}$ further gives $\langle RU^{T}x, RU^{T}y \rangle = \langle x, AA^{\dagger}y \rangle$.

By following a similar argument as above, we can show that $\|RU^{T}y\|_{2} = \|UU^{T}y \|_{2} = \|AA^{\dagger}y\|_{2}$, namely,
\begin{align}
    \|RU^{T}y\|_{2}^{2} = \langle RU^{T}y, RU^{T}y \rangle = \langle U^{T}y, U^{T}y \rangle = \langle UU^{T}y, UU^{T}y \rangle \nonumber,
\end{align}
where the last equality holds because $U^{T}U=I_{k}$. Therefore, we finally obtain the claimed equality
\begin{align}
    \langle RU^{T}x, \frac{RU^{T}y}{\|RU^{T}y\|_{2}} \rangle = \langle x, \frac{AA^{\dagger}y}{\|AA^{\dagger}y\|_{2}} \rangle \nonumber
\end{align}
which completes the proof.
\end{proof}

By using the result of Lemma~\ref{lemma:First}, we prove the following statement, which tells us the relationship between the effect of the dimensionality reduction in the Gram decomposition process and the found \net.

\begin{proposition}
Let $\mathtt{DB}_{N}$ be a set of facial images and $F$ be a template extractor.
For the matrices $A, U \in \mathbb{R}^{d \times k}$ and an unitary matrix $R \in \mathbb{R}^{k \times k}$ defined as Lemma~\ref{lemma:First}, let us denote $\{z_{1}, \dots, z_{Q}\} \subset \mathbb{R}^{k}$ as the solution of the MCP over the set $\{RU^{T}z: z = F_{T}(\mathfrak{b}); \mathfrak{b} \in \mathtt{DB}_{N}\}$ with a constraint that $\|z_{i}\|_{2}=1$ for $i=1,\dots,Q$. 
Then the set $\net := \{y_{1}, \dots, y_{Q} \} \subset \mathbb{R}^{d}$ defined as $y_{i} = UR^{T}z_{i}$ maximizes
\begin{align}
    \left | \left\{ \mathfrak{b} \in \mathtt{DB}_{N}: 
     F_{T}(\mathfrak{b}) \in \bigcup_{i=1}^{Q} \mathbb{B}_{\mathcal{X}} \left( \frac{AA^{\dagger}y_{i}}{\|AA^{\dagger} y_{i} \|_{2}}; \tau \right)
    \right\} \right|.\nonumber
\end{align}
That is, \net corresponds to the maximum coverage of $F_{T}(\mathtt{DB}_{N})$ when the solution space is constrained over the column space of $A$.
\end{proposition}

\begin{proof}
We first observe that $U$ has the full row rank; this implies that for any $z \in \mathbb{R}^{k}$, we can find $y \in \mathbb{R}^{d}$ such that $z = RU^{T}y$, or explicitly, $y = UR^{T}z$.

Then we focus on the metric ball over $\mathbb{R}^{k}$.
From the result of Lemma~\ref{lemma:First}, we can observe that, for $z \in \mathbb{R}^{k}$ and $y = UR^{T}z$
\begin{align}
    \langle RU^{T}x, \frac{z}{\|z\|_{2}} \rangle  > 1 - \tau \iff \langle x, \frac{AA^{\dagger}y}{\|AA^{\dagger}y\|_{2}} \rangle > 1 - \tau. \nonumber 
\end{align}
That is, if we consider the following optimization problem, whose solution is $\{z_{1}, \dots, z_{Q}\} \subset \mathbb{R}^{k}$,
\begin{align}
    \text{Maximize } &| \{ \mathfrak{b} \in \mathtt{DB}_{N}: \exists i \text{ s.t. }  \langle RU^{T}F_{T}(b), z \rangle > 1 - \tau  \}  | \nonumber \\
    \text{s.t. } &\|z_{i}\|_{2} =1 \text{ for } i=1,\dots, Q \nonumber,
\end{align}
we can deduce that the corresponding $\net := \{y_{1}, \cdots, y_{Q}\} \subset \mathbb{R}^{d}$ such that $y_{i} = UR^{T}z_{i}$, $i=1,\dots,Q$ becomes the solution of
\begin{align}
    \text{Maximize } &| \{ \mathfrak{b} \in \mathtt{DB}_{N}: \exists i \text{ s.t. } \langle F_{T}(b), \frac{AA^{\dagger}y_{i}}{\|AA^{\dagger}y_{i}\|_{2}} \rangle >  1- \tau  \}  | \nonumber \\
    \text{s.t. } &\|y_{i}\|_{2} =1 \text{ for } i=1,\dots, Q \nonumber.
\end{align}
Note that the constraints in the latter problem automatically holds in our choice of $y_{i}$ because $\|y_{i}\|_{2}^{2} = \langle UR^{T}z_{i}, UR^{T}z_{i} \rangle = \|z_{i}\|_{2}$.
Therefore, if we rewrite the above results in terms of the metric ball, then we finally obtain that \net maximizes
\begin{align}
    \left | \left\{ \mathfrak{b} \in \mathtt{DB}_{N}: 
     F_{T}(\mathfrak{b}) \in \bigcup_{i=1}^{Q} \mathbb{B}_{\mathcal{X}} \left( \frac{AA^{\dagger}y_{i}}{\|AA^{\dagger} y_{i} \|_{2}}; \tau \right)
    \right\} \right|.\nonumber
\end{align}
Since $AA^{\dagger} = UU^{T}$ is the projection operator onto the column space of $A$, the final center $\frac{AA^{\dagger}y_{i}}{\|AA^{\dagger}y_{i}\|_{2}}$ also lies within the column space of $A$, thus completing the proof.
\end{proof}

We remark that, since $RU^{T}x$ may not be a unit vector, the MCP solver over the $\mathbb{S}^{d-1}$ would not be well-defined.
Nevertheless, we observe that clustering-based methods that leverage the $\ell_{2}$ distance-based objective function, e.g., KMeans, provide the same solution to the given problem because
\begin{align}
    \| RU^{T}x-z\|_{2}^{2} = \|RU^{T}x\|_{2}^{2} + \|z\|_{2}^{2} - 2 \langle RU^{T}x, z\rangle \nonumber,
\end{align}
$\|z\|_{2}^{2}$ is the constant, and $\|RU^{T}x\|_{2}^{2}$ is independent of the choice of $z$.
That is, minimizing $\| RU^{T}x-z\|_{2}^{2}$ is equivalent to maximizing $\langle RU^{T}x, z\rangle$, which coincides with our analysis.

\balance
\section{K-Means-Greedy Hybrid MCP Solver}\label{sec:MCPDetail}

We provide the MCP solver we used during experiments.
Our MCP solver consists of two steps: the K-Means algorithm to find the candidates, and the Greedy algorithm to select the top $Q$ centers.
The motivation of such a hybrid approach stems from the difficulty of directly running the Greedy algorithm over the continuous space.
The candidates from K-Means serve as the discretization of the continuous space in terms of the local density, thus we can efficiently run the Greedy algorithm over the candidates.

For the K-Means algorithm, we utilize the $\ell_{2}$ distance objective function.
More precisely, for the given data points $\mathtt{DB}$ and the centroids $\{c_{1}, \dots, c_{K}\}$, the algorithm aims to minimize $\sum_{i=1}^{K}\sum_{x \in C_{i}}\|  x - c_{i} \|_{2}^{2}$, where $C_{i}$ denotes the cluster corresponding to $c_{i}$, i.e., 
\begin{align}
    C_{i} = \{x \in \mathtt{DB}: \arg\min_{i=1,\dots,K} \{ \|x-c_{i}\|_{2} \} = i  \}\nonumber.
\end{align}
Once the cluster is determined, we can minimize the objective function by selecting $c_{i} = \frac{1}{|C_{i}|}\sum_{x \in C_{i}}x$.
We repeat the process of determining the cluster and updating centers until convergence, or up to a fixed number of iterations.
Note that, if both $\mathtt{DB}$ and $\{c_{1}, \dots, c_{K}\}$ are sets of unit vectors, then the objective function is equivalent to minimizing cosine distance.

By running K-Means algorithm, we first find $\kappa Q$ candidates $\{w_{1}, \dots, w_{\kappa Q}\}$ for some $\kappa > 1$, which is a hyperparameter.
Among them, we select the final $Q$ centers to maximize the coverage of the given set $\mathtt{DB}$ in a greedy manner.
To this end, for each step, we first find a candidate $w_{i}$ such that its $\tau$-neighborhood $\mathbb{B}_{\mathbb{S}^{d-1}}(w_{i};\tau)$ maximally covers the points of $\mathtt{DB}$ that was not covered by previous candidates.
This can be realized by subtracting $\mathbb{B}_{\mathbb{S}^{d-1}}(w_{i};\tau)$ to $\mathtt{DB}$ for each selection of $w_{i}$. 
We repeat this process until we find all $Q$ candidates.
We describe the algorithm in Algorithm~\ref{alg:MCPSolver}.

\begin{algorithm}[t]
\caption{K-Means-Greedy Hybrid MCP Solver}\label{alg:MCPSolver}
\begin{algorithmic}[1]
\REQUIRE A set of points $\mathtt{DB} \subset \mathbb{S}^{d-1}$, the size of centroids $Q$, threshold $\tau$, multiplier $\kappa$, maximum iteration $n_{KM}$ of K-Means.
\ENSURE A set of points $\{z_{1}, \dots, z_{Q} \} \subset \mathbb{S}^{d-1}$.
\STATE Find $\kappa Q$ points $\{w_{1}, \dots, w_{\kappa Q}\} \gets \mathsf{KMeans}(\mathtt{DB}, \kappa Q, n_{KM})$.
\STATE Initialize an index set $\mathcal{I} \gets \emptyset$
\WHILE{$|\mathcal{I}| < Q$}
\STATE Find $i \gets \arg\max _{j \in \{1,\dots, \kappa Q \} \backslash \mathcal{I}} |\mathtt{DB} \cap \mathbb{B}_{\mathbb{S}^{d-1}}(w_{j}; \tau)|$
\STATE Append $\mathcal{I} \gets \mathcal{I} \cup \{i\}$ and update $\mathtt{DB} \gets \mathtt{DB} \backslash \mathbb{B}_{\mathbb{S}^{d-1}}(w_{i}; \tau)$
\ENDWHILE
\RETURN $\{z_{i} : i \in \mathcal{I}\}$.
\end{algorithmic}
\end{algorithm}

\end{document}